\documentclass[12pt]{article}
\usepackage[utf8]{inputenc}
\usepackage[T1]{fontenc}
\usepackage{lmodern}
\usepackage[english]{babel}
\usepackage{microtype}
\usepackage{datetime}
\usepackage{geometry}
\usepackage{setspace}
\usepackage{graphicx}
\usepackage{caption}
\usepackage{subcaption}
\usepackage{float}
\usepackage{booktabs}
\usepackage{tabularx}
\usepackage{ragged2e}
\usepackage{multirow}
\usepackage{amsmath,amssymb,amsthm,mathtools}
\usepackage[authoryear,round]{natbib}
\usepackage{enumitem}
\usepackage[dvipsnames]{xcolor}
\usepackage{xurl}
\usepackage{tikz}
\usetikzlibrary{calc,arrows.meta,positioning,patterns}
\usepackage{titling}
\usepackage[backref=page]{hyperref}
\usepackage[open,openlevel=2,atend,numbered,level=4]{bookmark}
\usepackage[nameinlink]{cleveref}

\setcitestyle{authoryear,round,aysep={,},citesep={;}}
\definecolor{deepblue}{RGB}{0,0,110}
\definecolor{darkgreen}{rgb}{0.0,0.5,0.0}
\hypersetup{
    pdftitle={First They Came for the Others: A Theory of Divide-and-Conquer},
    pdfauthor={Yeon-Koo Che, Jinyuqi Huang, Wooyoung Lim},
    colorlinks=true,
    linkcolor={deepblue},
    citecolor={deepblue},
    urlcolor={deepblue},
    hypertexnames=false,
    pdfpagemode=UseNone,
    pdfdisplaydoctitle=true
}

\renewcommand*{\backrefalt}[4]{%
    \ifcase #1
        No citation in the text.%
    \else
        [#2]%
    \fi}
\setlist[enumerate]{leftmargin=1.5cm,rightmargin=0.5cm,noitemsep,topsep=2pt}
\allowdisplaybreaks
\makeatletter
\newcommand\blfootnote[1]{%
  \begingroup
  \renewcommand\thefootnote{}\footnote{#1}%
  \addtocounter{footnote}{-1}%
  \endgroup
}
\makeatother

\theoremstyle{definition}
\newtheorem{assumption}{Assumption}

\newtheorem{definition}{Definition}

\newtheorem{theorem}{Theorem}
\newtheorem{proposition}{Proposition}
\newtheorem{lemma}{Lemma}

\newtheorem{corollary}{Corollary}

\crefname{assumption}{Assumption}{Assumptions}
\Crefname{assumption}{Assumption}{Assumptions}
\crefname{example}{Example}{Example}
\Crefname{example}{Example}{Example}
\crefname{definition}{Definition}{Definitions}
\Crefname{definition}{Definition}{Definitions}
\crefname{remark}{Remark}{Remarks}
\Crefname{remark}{Remark}{Remarks}
\crefname{theorem}{Theorem}{Theorems}
\Crefname{theorem}{Theorem}{Theorems}
\crefname{proposition}{Proposition}{Propositions}
\Crefname{proposition}{Proposition}{Propositions}
\crefname{lemma}{Lemma}{Lemmas}
\Crefname{lemma}{Lemma}{Lemmas}
\crefname{observation}{Observation}{Observations}
\Crefname{observation}{Observation}{Observations}
\crefname{corollary}{Corollary}{Corollaries}
\Crefname{corollary}{Corollary}{Corollaries}
\crefname{figure}{Figure}{Figures}
\Crefname{figure}{Figure}{Figures}
\crefname{table}{Table}{Tables}
\Crefname{table}{Table}{Tables}
\crefname{section}{Section}{Sections}
\Crefname{section}{Section}{Sections}
\crefname{appendix}{Appendix}{Appendices}
\Crefname{appendix}{Appendix}{Appendices}
\crefname{equation}{}{}
\Crefname{equation}{}{}
\crefformat{equation}{#2(#1)#3}
\Crefformat{equation}{#2(#1)#3}
\crefrangeformat{equation}{#3(#1)#4--#5(#2)#6}
\Crefrangeformat{equation}{#3(#1)#4--#5(#2)#6}

\newcommand{\Pbb}{\Pr}

\newcommand{\DC}{D\&C}
\newcommand{\Vone}{V_1}
\newcommand{\Vtwo}{V_2}
\newcommand{\A}{A}
\newcommand{\hatc}{\hat c}

\tikzset{every picture/.style={font=\scriptsize}}
\definecolor{dcblue}{RGB}{77,145,214}
\definecolor{dcred}{RGB}{207,76,66}
\definecolor{dcorange}{RGB}{235,157,58}
\definecolor{dcgreen}{RGB}{106,168,79}
\definecolor{dcpurple}{RGB}{126,87,194}

\definecolor{colCR}{RGB}{210,70,60}          
\definecolor{colDCone}{RGB}{255,210,80}      
\definecolor{colDCboth}{RGB}{240,165,40}     
\definecolor{colHypL}{RGB}{205,205,205}      
\definecolor{colHypD}{RGB}{170,170,170}      

\thanksmarkseries{arabic}
\makeatletter
\renewcommand*{\@fnsymbol}[1]{\ifcase#1\or *\else\@arabic{#1}\fi}
\makeatother

\title{\bf First They Came for the Others:\\ A Theory of Divide-and-Conquer}
\author{
    Yeon-Koo Che\thanks{Department of Economics, Columbia University, USA. Email: \href{mailto:yeonkooche@gmail.com}{\texttt{yeonkooche@gmail.com}}.}\qquad
    Jinyuqi Huang\thanks{Department of Economics, The Hong Kong University of Science and Technology, Hong Kong. Email: \href{mailto:jhuangde@connect.ust.hk}{\texttt{jhuangde@connect.ust.hk}}.}\qquad
    Wooyoung Lim\thanks{Department of Economics, The Hong Kong University of Science and Technology, Hong Kong. Email: \href{mailto:wooyoung@ust.hk}{\texttt{wooyoung@ust.hk}}.}
}
\date{\normalsize This version: October 3, 2026}

\begin{document}
\setstretch{1.22}

\maketitle

\blfootnote{We are grateful to Matteo Camboni, H\"ulya Eraslan, Marina Halac, Sam Kapon, Andrew Koh, Andrew Kosenko, Elliot Lipnowski, Qingmin Liu, Alessandro Pavan, Nicola Persico, Carlo Prato, and Kathryn Spier for helpful comments and suggestions. This work was supported by the Ministry of Education of the Republic of Korea and the National Research Foundation of Korea (NRF-2024S1A5A2A03038509).}

\begin{abstract}
Divide-and-conquer tactics often succeed not through mechanical coordination failures, but through inferential friction regarding an aggressor's underlying intent. When an attacker strikes a first target, bystanders must infer whether the assault represents a localized grievance or a systemic campaign. If the attack is rationally interpreted as particularized, bystanders abstain, prompting the isolated victim to surrender. We demonstrate how higher attack costs and lower correlation between victims' fates facilitate this division. We then study how behavioral responses, rhetoric, treaty commitments, and downstream defense networks modify this inference.
\end{abstract}

 \begin{flushleft}
\textit{First they came for the socialists, and I did not speak out---because I was not a socialist.\\
Then they came for the trade unionists, and I did not speak out---because I was not a trade unionist.\\
Then they came for the Jews, and I did not speak out---because I was not a Jew.\\
Then they came for me---and there was no one left to speak for me.}\\
--- Martin Niemöller
\end{flushleft}
\vspace{0.5cm}

\section{Introduction}

History is replete with aggressors who consolidate power or territory sequentially rather
than in a single overwhelming strike. In hindsight, the failure of potential victims to
mount a united front often appears puzzling. Yet, at the moment of the first attack, the
attacker's terminal ambitions are rarely obvious. A first attack may reveal a broad
campaign, but it may also reflect a localized grievance, a disputed border, a particular
ethnic claim, or a target-specific opportunity. The bystander's problem is therefore not
only whether collective resistance would be valuable if all victims were targeted. It is
whether the first attack is informative enough to justify treating the victims' fates as
linked.

The 1938 Munich Agreement is a canonical example. Hitler framed the Sudetenland as a
limited grievance involving ethnic Germans, and the resulting concession delayed a common
recognition of the broader Nazi threat.\footnote{Japan's incremental domination of Korea, from the 1905 Protectorate
Treaty to full annexation in 1910, offers a related case. Early encroachments could be
framed as localized regional adjustments rather than the beginning of full annexation
\citep{BritannicaKoreaJapaneseRule}.} More recently, Russia's 2014 annexation of Crimea
was presented through claims of local separatism and historical exceptionalism, years before
the 2022 full-scale invasion of Ukraine. The same logic appears in domestic politics, when authoritarians move against their rivals and democratic institutions one at a time.  M\'aty\'as R\'akosi described the elimination of opposition in Hungary as
``salami tactics'': opponents were isolated and removed slice by slice, while each remaining
group had reason to hope that it was not next \citep{TimeRakosi1947,TNSRSalami2021}.\footnote{
Stalin's sequential elimination of rivals had the same structure: first aligning with
Zinoviev and Kamenev against Trotsky, then with the right against Zinoviev and Kamenev,
and finally against the right itself \citep{BritannicaGreatPurge}.}

Standard models of conflict often attribute the failure of collective resistance to
conflicting interests, free-rider problems, or mechanical coordination failures. These
frictions are important, but they leave out a distinct constraint: uncertainty about the
aggressor's true intent. When an attack on one victim is observed, a bystander must infer
whether the attack is particularized or systemic. Divide-and-conquer succeeds when the
localized interpretation is sufficiently plausible. The bystander is not physically prevented
from joining, and she need not misunderstand the value of collective defense. She stays out
because the first attack may not be about her.

This paper develops a theory of divide-and-conquer based on this inferential friction. An
attacker privately observes the benefits of conquering two potential victims, \(\Vone\) and
\(\Vtwo\). Joint resistance is effective, and the timing allows the bystander to join after
observing the first victim's resistance decision. Thus, if both victims know that they are
targets, they resist jointly. The only reason the second victim stays out is uncertainty
about whether the first attack reveals a threat to her. In the baseline model, \(\Vone\) is
commonly understood to be the more natural first target.

The baseline model has two canonical equilibria. In a divide-and-conquer equilibrium, an
attack on \(\Vone\) is sufficiently consistent with a target-specific motive that \(\Vtwo\)'s
posterior risk remains below her belief threshold. She stays out; \(\Vone\) then surrenders,
because fighting alone is worse than surrender; and the attacker can proceed sequentially.
In a collective-resistance equilibrium, the same first attack is sufficiently informative of
a broader campaign; \(\Vtwo\) joins, \(\Vone\) fights, and the attacker attacks only when a
two-front war is profitable. The equilibrium comparison is governed by a single posterior
question: after seeing the first attack, how likely is the bystander to be next?
When these two pure equilibria coexist, a mixed equilibrium also exists in which an attacker hoping for acquiescence encounters collective resistance with positive probability.

This informational architecture yields a set of striking comparative statics. Paradoxically,
a higher cost of attack can make divide-and-conquer easier to sustain. When aggression is
viewed as norm-breaking or ``unthinkable,'' an observed attack is more readily rationalized
as exceptional: the bystander reasons that the attacker must have a highly particularized
reason for incurring such a large moral, diplomatic, or reputational cost. In a lower-cost
environment, by contrast, a first attack has less exceptional force and is more easily read as
the opening move of a broader campaign. Thus, the same cost that discourages aggression ex
ante can, conditional on an attack being observed, make the first strike look less systemic.
The other comparative statics operate through the same posterior logic: lower correlation
among the victims' fates weakens the Niem\"oller-style inference, while greater perceived
``specialness'' of the first target makes division easier.

We then study richer environments that influence this inference. Behavioral responses
affect the bystander's willingness to act on danger: wishful thinking permits a reassuring
posterior, while magical thinking makes neutrality appear protective. Both help sustain the divide-and-conquer equilibrium. Rhetoric and
propaganda affect the information on which the bystander conditions: particularizing
messages sent by the attacker frame the first attack as exceptional and so facilitate division; by contrast, systemic messages sent by early victims link the attack to a broader
campaign, thus helping to mobilize joint resistance with other victims.  Treaty commitments show that institutions can either close or widen the wedge between the victims:
collective-defense treaties make inaction costly, whereas neutrality or non-aggression
arrangements may reassure the bystander and expose the first victim. With more than two
victims, the protection of later victims can make an intermediate victim less willing to defend the first one. The common takeaway is that divide-and-conquer turns on how actions,
rhetoric, and institutions shape the bystander's interpretation of the first attack.

The paper contributes to the literature on conflict, collective action, contracting with externalities, and signaling. Models of discriminatory multi-agent contracting show how a principal can use different offers to overcome agents' coordination or guarantee a desired outcome \citep{Segal2003,Winter2004,HalacKremerWinter2020}. Our mechanism instead concerns inference about the scope of aggression: the first attack changes the bystander's assessment of whether she is also a target. The model complements rationalist explanations for war and commitment problems \citep{Fearon1995,Powell2006}, and models of appeasement and deterrence \citep{GurantzHirsch2017}.

The remainder of the paper is organized as follows. \Cref{sec:model} introduces the model.
\Cref{sec:baseline} characterizes the baseline equilibria and comparative statics.
\Cref{sec:extensions} studies behavioral responses, rhetoric, treaties, and richer victim
networks. \Cref{sec:conclusion} concludes.

\section{Model}\label{sec:model}

\subsection{Players, payoffs, and timeline}

There is one attacker, denoted by \(\A\), and two potential victims, \(\Vone\) and \(\Vtwo\). The attacker privately observes her benefits \((b_1,b_2)\) from obtaining control or concessions from the two victims. These benefits combine common ambitions with target-specific stakes. A territorial claim or strategic asset can make the first victim valuable even when the second offers little benefit; broader ambitions can make both valuable. The circumstances may be public, while the value the attacker places on them is private. The victims know the joint distribution but not the realized benefits. The attacker's type is drawn from
\[
  B:=\{(b_1,b_2):0\leq b_2\leq b_1\leq 1\},
\]
according to distribution \(\mu\), with density \(f\) continuous and strictly positive on the interior of \(B\). The support captures environments in which all parties recognize that \(\Vone\) is the attacker's higher-stakes target. For example, \(\A\) may have a territorial claim against \(\Vone\) but not against \(\Vtwo\); \(\Vone\) may be geographically closer; or \(\Vone\) may have particular strategic or symbolic value, as Crimea did for Russia. The cost of attacking a victim is \(c>0\), which is common knowledge. We study \(0<c<1\); for \(c\ge1\), no positive mass of types attacks under any response. Besides material costs, \(c\) may include reputational, diplomatic, and moral costs of violating accepted norms. Attacking two victims costs \(2c\). The attacker's payoff from not attacking is normalized to zero.

We write \(F_i\) for the marginal distribution function of \(b_i\).   The general model permits dependence and non-uniform benefits.
A natural benchmark, which we use for some comparative statics, takes the attacker's values for the two victims to be independent draws from a common distribution. The values are her private information, but it is commonly known which victim she values more. With that victim labeled \(\Vone\), the benefits are the ordered values of the two draws.

\begin{definition}[Ordered independent draws]\label{def:ordered}
The benefits satisfy the \emph{ordered-independent-draw condition} if
\[
b_1=\max\{X_1,X_2\},\qquad b_2=\min\{X_1,X_2\},
\]
where \(X_1,X_2\) are independent draws from a common continuous distribution \(F\) on \([0,1]\), with positive density on \((0,1)\).
\end{definition}

A special case is the uniform distribution on the triangle \(B\), which is obtained when \(F\) is uniform, \(F(x)=x\).

A victim who is conquered suffers loss \(\ell>0\). This loss may represent territorial loss, loss of political autonomy, loss of security, exclusion from markets, or other consequences of conquest. A victim who fights incurs cost \(d>0\). If a victim \(i\) fights alone, the attacker wins \(b_i\) with probability \(p_1\). If both victims fight together, the attacker wins \(b_1+b_2\) with probability \(p_2\), where
\[
  0<p_2<p_1<1.
\]
(The subscript on \(p\) counts the number of victims who fight, not their identities.) Thus the victim's expected loss from fighting alone is \(p_1\ell+d\), while the expected loss from fighting together is \(p_2\ell+d\). We maintain the following assumption.

\begin{assumption}[Resistance is effective only if coordinated]\label{ass:victims}
\[
  p_2\ell+d<\ell<p_1\ell+d.
\]
\end{assumption}

The left inequality says that joint fighting is better for a victim than surrender. The right inequality says that fighting alone is worse than surrender.

The game proceeds as follows.

\begin{enumerate}[label=\textit{Stage \arabic*.},leftmargin=*]
\item \(\A\) observes \((b_1,b_2)\) and chooses whether to attack \(\Vone\). If \(\A\) does not attack, the game ends.

\item If \(\A\) attacks \(\Vone\), \(\Vone\) chooses whether to surrender or fight. If \(\Vone\) fights, \(\Vtwo\) observes the attack and \(\Vone\)'s resistance decision, and chooses whether to join the fight. If \(\Vtwo\) joins, \(\A\) fights both victims and the game ends. If \(\Vtwo\) does not join, \(\A\) fights \(\Vone\) alone.

\item Unless the game has ended in a joint fight, \(\A\) chooses whether to attack \(\Vtwo\). If \(\A\) attacks \(\Vtwo\), \(\Vtwo\) chooses whether to surrender or fight alone.
\end{enumerate}

The corresponding attacker payoffs are as follows. If \(\A\) attacks \(\Vone\), \(\Vone\) surrenders, and \(\A\) does not attack \(\Vtwo\), the attacker receives $ b_1-c.$
 If \(\A\) attacks \(\Vone\), \(\Vone\) surrenders, and \(\A\) later attacks \(\Vtwo\), who also surrenders, the attacker receives
 $b_1-c+b_2-c.$  If a victim \(i\) fights alone, the attacker's expected payoff from that fight is $p_1b_i-c.$  
If \(\Vone\) and \(\Vtwo\) fight together, the attacker's expected payoff is
  $p_2(b_1+b_2)-2c.$

Two comments on the extensive form are useful. First, the fighting decisions are
sequential in order to focus on the inferential friction. If the victims chose simultaneously
whether to fight, the model would contain a standard coordination problem: each victim
might want to fight if the other fights, but fear being left alone. We abstract from that
force by allowing the bystander to observe the first victim's resistance decision before
deciding whether to join. Thus, if the bystander knows that she is also threatened, she can
join; anticipating this, the attacked victim fights. Any failure of collective resistance must
therefore come from uncertainty about the attacker's target scope, not from simultaneous-move miscoordination.

Second, the extensive form is streamlined for exposition. We could allow \(\A\) either to
attack both victims simultaneously or to begin by attacking \(\Vtwo\). If \(\A\) attacked both
victims simultaneously, both victims would know that the attack is systemic; under
\Cref{ass:victims}, they would fight jointly, and \(\A\)'s payoff would be
\(p_2(b_1+b_2)-2c\). Attacking \(\Vone\) first gives \(\A\) at least this payoff, whatever the victims' response: if the first attack triggers
joint resistance, \(\A\) obtains the same two-front payoff; if it does not, \(\A\) preserves the
option to attack \(\Vtwo\) only when doing so is profitable. If \(\A\) began with \(\Vtwo\) instead, the attack would reveal, since \(b_2\le b_1\), that \(\Vone\) is also profitable;
hence, if \(\Vtwo\) fought, \(\Vone\) would join. The resulting two-front payoff is no higher than what \(\A\) obtains by attacking \(\Vone\) first. \Cref{app:v2first} provides the details.\footnote{We assume in the main text that, if \(\Vtwo\) stays out after an off-path lone fight by \(\Vone\), \(\A\) can later attack \(\Vtwo\) regardless of the outcome of that fight. This is for expositional ease. \Cref{oa:reachability} considers an alternative specification in which \(\A\) can attack \(\Vtwo\) only if she first defeats \(\Vone\) in the lone fight. Divide-and-conquer is then easier to sustain because staying out exposes \(\Vtwo\) to a future attack only when \(\A\) wins that fight.}

\subsection{Beliefs and equilibrium}\label{sec:equilibrium}

We use Perfect Bayesian Equilibrium with a local consistency restriction based on the ``no signaling what you don't know'' principle \citep{FudenbergTirole1991}. A victim does not observe \(\A\)'s type, so an unexpected action by that victim should not by itself transmit information about \(\A\)'s type. The bystander still updates from \(\A\)'s observed attack and from \(\A\)'s equilibrium strategy; she simply does not update further from the other victim's off-path action.

\begin{definition}[Regular PBE]\label{def:regular}
A \emph{regular PBE} is a PBE satisfying the following consistency requirement. At any information set reached only because a victim deviates from her prescribed response, the other victim's posterior belief about \(\A\)'s type is obtained from Bayes' rule using \(\A\)'s observed action and \(\A\)'s equilibrium strategy, and is not updated further on the basis of the deviating victim's unexpected action.
\end{definition}

For example, suppose that, in equilibrium, \(\A\) attacks \(\Vone\) on a measurable set \(E\subset B\) with \(\mu(E)>0\), and suppose \(\Vone\)'s prescribed response after such an attack is surrender. If \(\Vone\) unexpectedly fights, regularity requires \(\Vtwo\)'s belief \(\mu_2\) to satisfy
\[
  \mu_2(\,\cdot\mid \A\text{ attacks }\Vone,\ \Vone\text{ fights})
  =
  \mu_2(\,\cdot\mid \A\text{ attacks }\Vone)
  =
  \mu(\,\cdot\mid E).
\]
The bystander updates from the fact that \(\A\) attacked \(\Vone\), but not from \(\Vone\)'s off-path fight. This is the only off-path belief restriction needed for the baseline analysis. The requirement is close in spirit to sequential equilibrium \citep{KrepsWilson1982}, but weaker because it disciplines only the local off-path belief needed here.  Henceforth, unless otherwise stated, ``equilibrium'' means regular PBE.

\subsection{Full-information benchmark}\label{sec:full-info}

Suppose \((b_1,b_2,c)\) is commonly known. If \(b_2\geq c\), then after \(\Vone\) is conquered, \(\A\) would find it profitable to attack \(\Vtwo\). If \(\Vtwo\) stays out of \(\Vone\)'s fight, she is later conquered and suffers loss \(\ell\). If she joins \(\Vone\)'s fight, her loss is \(p_2\ell+d\). By \Cref{ass:victims}, \(\Vtwo\) joins. Anticipating this, \(\Vone\) fights.

Thus, under full information, successful sequential conquest of both victims through passivity is impossible. Whenever \(\Vtwo\) is known to be a future target, the victims resist jointly. The attacker reaches \(\Vtwo\) only when she is willing to fight a two-front war and when \(\Vtwo\) is indeed a profitable target, i.e., only for types satisfying
\[
  b_2\geq c
  \quad\text{and}\quad
  p_2(b_1+b_2)\geq 2c.
\]
Successful sequential conquest of both victims through passivity, therefore, requires incomplete information about the attacker's target scope.

\section{Baseline Analysis}\label{sec:baseline}

\subsection{Characterization of equilibria}\label{sec:reading}

The pure-strategy analysis has two possible equilibrium outcomes. In a {\bf Divide-and-Conquer (\DC)} outcome, \(\Vtwo\) stays out if \(\Vone\) fights; anticipating that resistance would be solitary, \(\Vone\) surrenders. In a {\bf Collective-resistance} outcome, \(\Vtwo\) joins if \(\Vone\) resists; anticipating help, \(\Vone\) fights. We construct the two outcomes in turn.

Consider first a \DC\ candidate. If \(\Vtwo\) would stay out after a fight by \(\Vone\), then \(\Vone\)'s resistance would be a lone fight and \(\Vone\) surrenders by \Cref{ass:victims}. Given surrender by \(\Vone\), \(\A\)'s payoff from attacking \(\Vone\) is
\[
  b_1-c+\max\{b_2-c,0\}.
\]
On the support \(B\), this payoff is nonnegative if and only if \(b_1\ge c\). If \(\Vone\) surrenders, \(\A\) later attacks \(\Vtwo\) if and only if \(b_2\ge c\). Thus, after observing an attack on \(\Vone\), regularity gives the posterior
\begin{equation}\label{eq:phi}
  \phi(c):=\Pr(b_2\ge c\mid b_1\ge c).
\end{equation}
This is the probability that \(\A\) also finds \(\Vtwo\) profitable, conditional on the fact that attacking \(\Vone\) is profitable.

\begin{figure}[!htbp]
\centering
\begin{tikzpicture}[scale=6.0, font={\fontsize{11pt}{13pt}\selectfont}]
  \coordinate (O) at (0,0);
  \coordinate (A) at (1,0);
  \coordinate (B) at (1,1);

  \begin{scope}
    \clip (O)--(B)--(A)--cycle;
    \fill[colDCone!50] (.42,0) rectangle (1,1);
    \fill[colDCboth!55] (0,.42) rectangle (1,1);
  \end{scope}

  \draw[->] (0,0) -- (1.08,0) node[right] {$b_1$};
  \draw[->] (0,0) -- (0,1.08) node[above] {$b_2$};

  \draw (O) -- (B) -- (A);

  \draw[densely dotted] (0,0)--(1,1) node[pos=.36,above left] {$45^\circ$};
  \draw[dashed] (.42,0)--(.42,.42)--(1,.42);
  
  \node[below] at (.42,0) {$c$};
  \node[left]  at (0,.42) {$c$};
  
  \node[align=center,black]  at (.71,.21) {Attack $V_1$ only};
  \node[align=center,black] at (.71,.50) {Both targets};
\end{tikzpicture}
\caption{The posterior \(\phi(c)\). The horizontal axis measures \(b_1\) and the vertical axis \(b_2\). The triangle below the $45^\circ$ line denotes the baseline support $B$. The full shaded region (light and dark yellow combined) is \(\{b_1\ge c\}\cap B\), the set of types who attack \(V_1\) if \(V_1\) is expected to surrender. The dark yellow region is \(\{b_2\ge c\}\cap \{b_1\ge c\}\cap B\), the subset that also finds \(V_2\) profitable. The bystander compares the probability of the dark yellow region with that of the whole shaded region.}
\label{fig:phi}
\end{figure}

\Cref{fig:phi} gives the geometry. The observation that \(\Vone\) was attacked places the type in the region \(b_1\ge c\). Within that region, \(\Vtwo\) is next precisely for types with \(b_2\ge c\); \(\phi(c)\) is the probability that \((b_1,b_2)\) lies in the darker-shaded triangle, conditional on its lying in the shaded trapezoid.

If \(\Vtwo\) stays out, her expected loss is \(\phi(c)\ell\). If she joins \(\Vone\)'s fight, her loss is \(p_2\ell+d\). Hence staying out is optimal exactly when
\[
  \phi(c)\ell\le p_2\ell+d,
\]
or equivalently
\begin{equation}\label{eq:kappa}
  \phi(c)\le \kappa:=\frac{p_2\ell+d}{\ell}.
\end{equation}
By \Cref{ass:victims}, \(\kappa\in(0,1)\). The number \(\kappa\) is the bystander's \emph{belief threshold}: it is the posterior probability of being attacked next at which she is indifferent between staying out and joining. A higher \(\kappa\) means that joint resistance is relatively costly or ineffective, so the bystander tolerates a larger posterior danger. A lower \(\kappa\) means that a smaller posterior danger triggers collective defense.

Now consider a collective-resistance candidate. If \(\Vtwo\) would join after a fight by \(\Vone\), then \(\Vone\) fights when attacked. Anticipating joint resistance, \(\A\) attacks only when a two-front war is profitable:
\[
  p_2(b_1+b_2)\ge 2c.
\]
Let
\[
  J(c):=\{(b_1,b_2)\in B:p_2(b_1+b_2)\ge 2c\}
\]
denote this attack set. If \(\mu(J(c))>0\), define
\begin{equation}\label{eq:psi}
  \psi(c):=\Pr(b_2\ge c\mid (b_1,b_2)\in J(c)),
\end{equation}
with the convention that \(\psi(c)=1\) when \(J(c)\) has zero probability. If \(\Vtwo\) stayed out after an attack drawn from this set, her expected loss would be \(\psi(c)\ell\), while joining costs \(p_2\ell+d=\kappa\ell\). Hence joining is optimal exactly when
\[
  \psi(c)\ge \kappa.
\]

\begin{assumption}[Positive selection into collective resistance] \label{ass:positive-selection}
For every \(c\), \(\psi(c)\ge \phi(c)\).
\end{assumption}

\Cref{ass:positive-selection} says that \(\Vtwo\) is at least as likely to be next when \(\A\) attacks anticipating joint resistance as when \(\A\) attacks anticipating that \(\Vone\) will surrender. This is natural: an attacker who anticipates joint resistance attacks only if \(p_2(b_1+b_2)\ge 2c\), that is, only if the total benefit is high, and a high total benefit is likely to be associated with a high value of \(b_2\). The condition clearly holds when \(p_2\) is small or \(c\) is large: if \(p_2<2c/(1+c)\), then \(J(c)\subseteq\{b_2\ge c\}\) on \(B\); whenever \(\mu(J(c))>0\), this gives \(\psi(c)=1\ge \phi(c)\). \Cref{fig:two-front} shows the geometry for a small \(p_2\).

\begin{figure}[!htbp]
\centering
\begin{tikzpicture}[scale=6.0, font={\fontsize{11pt}{13pt}\selectfont}]
  \draw[->] (0,0) -- (1.08,0) node[right] {$b_1$};
  \draw[->] (0,0) -- (0,1.08) node[above] {$b_2$};
  \node[below] at (.28,0) {$c$};
  \node[left] at (0,.28) {$c$};
  \begin{scope}
    \clip (0,0)--(1,0)--(1,1)--cycle;
    \fill[gray!25] (.28,.28) rectangle (1,1);
    \fill[colCR!45] (.72,.72)--(1,.44)--(1,1)--cycle;
  \end{scope}
  \draw[very thick,colCR!85!black] (.72,.72)--(1,.44);
  \draw[dashed] (.28,0)--(.28,.28)--(1,.28);
  \draw (0,0)--(1,1)--(1,0);
  \node[align=center,black] (tf) at (.40,.87) {Two-front attack};
  \draw[->,thin] (tf.east) -- (.885,.80);
  \node[align=center] at (.66,.45) {$b_2\ge c$};
  \node[align=center] at (.66,.16) {$b_2<c$};
\end{tikzpicture}
\caption{The two-front attack set. The red region is \(J(c)=\{p_2(b_1+b_2)\ge 2c\}\) for a small \(p_2\). A demanding two-front-war threshold selects types near the high-benefit corner and is therefore informative that \(\Vtwo\) is also at risk.}
\label{fig:two-front}
\end{figure}

A primitive sufficient condition is given below in \Cref{ass:association}. In particular, positive selection holds for the ordered values of two independent draws from any common continuous distribution (\Cref{oa:orderstat}). Without positive selection, the two separate existence conditions in \Cref{thm:baseline} below still apply, but the three-region classification need not: if \(\psi(c)<\kappa<\phi(c)\), neither canonical pure equilibrium exists.

With this assumption, the next theorem gives a clean classification.

\begin{theorem}[Pure equilibria in the baseline model]\label{thm:baseline}
Fix \(c\) and suppose \Cref{ass:victims} holds. Up to payoff-irrelevant tie-breaking at cutoff types, every pure regular equilibrium is one of the following two kinds.

\begin{enumerate}[label=(\roman*),leftmargin=2.7em]
\item \textbf{\DC\ equilibrium.} A \DC\ equilibrium exists if and only if
\begin{equation}\label{eq:DC-condition}
  \phi(c)\le \kappa.
\end{equation}
In this equilibrium, \(\A\) attacks \(\Vone\) if and only if \(b_1\ge c\). If \(\Vone\) surrenders, \(\A\) subsequently attacks \(\Vtwo\) if and only if \(b_2\ge c\). Both victims surrender. If \(\Vone\) were to fight, \(\Vtwo\) would stay out.

\item \textbf{Collective-resistance equilibrium.} A collective-resistance equilibrium exists if and only if
\begin{equation}\label{eq:CR-condition}
  \psi(c)\ge \kappa.
\end{equation}
If \(\mu(J(c))>0\), \(\Vone\) fights when attacked and \(\Vtwo\) joins; anticipating joint resistance, \(\A\) attacks if and only if \((b_1,b_2)\in J(c)\). If \(\mu(J(c))=0\), collective resistance fully deters attack.
\end{enumerate}

If \Cref{ass:positive-selection} also holds, the equilibrium set is classified as follows:
\[
\begin{array}{ll}
\phi(c)>\kappa: & \text{only collective resistance exists in pure strategies;}\\[3pt]
\phi(c)\le \kappa\le \psi(c): & \text{both \DC\ and collective resistance exist;}\\[3pt]
\psi(c)<\kappa: & \text{only \DC\ exists in pure strategies.}
\end{array}
\]
\end{theorem}

The proof is in \Cref{app:proofs}. The construction above gives the main incentive logic. In a \DC\ equilibrium, the attack on \(\Vone\) tells \(\Vtwo\) that \(b_1\ge c\), but not whether \(b_2\ge c\). Staying out gives expected loss \(\phi(c)\ell\), while joining gives loss \(p_2\ell+d\). Thus \DC\ is sequentially rational exactly when \(\phi(c)\le\kappa\). Given that \(\Vtwo\) stays out, \(\Vone\)'s resistance would be solitary, so \(\Vone\) surrenders by \Cref{ass:victims}.

In a collective-resistance equilibrium, \(\Vtwo\) conditions on the attack set \(J(c)\), unless the threat of joint resistance fully deters attack. Staying out gives expected loss \(\psi(c)\ell\), while joining gives loss \(p_2\ell+d\). Thus joining is sequentially rational exactly when \(\psi(c)\ge\kappa\). Knowing that help will come, \(\Vone\) fights. Anticipating joint resistance, \(\A\) attacks only when a two-front war is profitable.

The theorem also explains why there are no other pure regular equilibria. After the history in which \(\A\) attacks \(\Vone\) and \(\Vone\) fights, \(\Vtwo\)'s pure response is either to stay out or to join. The first case gives the \DC\ incentives above; the second gives the collective-resistance incentives above.

\Cref{thm:baseline} concerns pure strategies. The next condition rules out mixed equilibria as well when \DC\ is not sustainable.

\begin{assumption}[Association among potential attackers]\label{ass:association}
For every \(0<c<1\) and cutoff \(x\), \(\Pr(b_1\ge x\mid b_2=y,\ b_1\ge c)\) is nondecreasing in \(y\).
\end{assumption}

Among types who would attack a surrendering first victim (those with \(b_1\ge c\)), a higher benefit from the second victim is associated with a stochastically higher benefit from the first. This permits non-uniform marginals and correlation. It holds for ordered independent draws from any common distribution and, more generally, for affiliated benefits; \Cref{oa:mixed} gives sufficient density conditions.

\begin{corollary}[Uniqueness when division is not sustainable]\label{cor:unique-CR}
Under \Cref{ass:victims,ass:association}, if \(\phi(c)>\kappa\), collective resistance is the unique equilibrium outcome, even when mixed strategies are allowed. The initial attack set is \(J(c)\); when \(c\ge p_2\), attacks are fully deterred.
\end{corollary}

The association condition also implies \Cref{ass:positive-selection}. The corollary strengthens the first case of \Cref{thm:baseline}: when \DC\ is not sustainable, collective resistance is the only equilibrium outcome even if the victims may randomize. Its proof is in \Cref{oa:mixed}.

\subsection{Interpretation and multiplicity}\label{sec:interp}

The theorem clarifies the informational nature of divide-and-conquer. Under full information, if \(\Vtwo\) knows she is next, she joins and \(\Vone\) fights. Under incomplete information, the same first attack can be read in two ways. It can reveal a systemic campaign, in which case the victims unite. Or it can be rationalized as a special action against \(\Vone\), in which case \(\Vtwo\) stays out and \(\Vone\) surrenders.

This logic of the \DC\ equilibrium is close to the lesson commonly associated with Martin Niemöller's ``First they came'' formulation. The United States Holocaust Memorial Museum describes the text as a reflection on silence and complicity in the face of Nazi persecution \citep{USHMMNiemoller2023}. Our model does not require bystanders to be indifferent or irrational. It identifies a Bayesian wedge: after the first attack, the bystander may think, ``perhaps they came for them, not for me.'' The attacker's advantage lies in making that interpretation credible.

When both equilibria exist, they represent two different but internally consistent ways for victims to coordinate their beliefs. In the collective-resistance equilibrium, an attack is interpreted as coming from types willing to face joint resistance; this interpretation makes the attack sufficiently alarming to compel the victims to fight jointly.  In the \DC\ equilibrium, an attack is interpreted against the broader pool of types that merely find \(\Vone\) profitable; this interpretation leaves enough room for the hope that the attack is localized, leading to division. 

When the two pure equilibria coexist, there is also a mixed equilibrium, in which the response to an attack is uncertain. Let \(\alpha\) be the probability that \(\Vone\) fights after an attack and \(\beta\) the probability that \(\Vtwo\) joins after observing that fight. Write \(B_{\alpha,\beta}(c)\) for the resulting initial attack set, derived in \Cref{oa:mixed}.

\begin{proposition}[Uncertain collective resistance]\label{prop:mixed}
If \(\phi(c)<\kappa<\psi(c)\), a mixed equilibrium exists with positive attack probability and collective resistance occurring with probability strictly between zero and one. It satisfies
\[
\Pr(b_2\ge c\mid B_{\alpha,\beta}(c))=\kappa,
\qquad
J(c)\subseteq B_{\alpha,\beta}(c)\subseteq\{b_1\ge c\},
\]
with both inclusions strict in probability.
Conditional on an initial attack and no joint resistance, the probability of a subsequent attack on \(\Vtwo\) is exactly \(\kappa\).
\end{proposition}

The proof is in \Cref{oa:mixed}; the logic is as follows. The more likely the victims are to resist, the less attractive an attack becomes, so the attack set shrinks from \(\{b_1\ge c\}\), its value when surrender is certain, toward \(J(c)\), its value when joint resistance is certain. The bystander's posterior moves accordingly from \(\phi(c)<\kappa\) to \(\psi(c)>\kappa\). At an intermediate likelihood of resistance it equals \(\kappa\), which makes the bystander willing to randomize; the first victim, for her part, fights only if the bystander is sufficiently likely to join.

In this equilibrium, the attacker faces uncertainty about whether the victims will resist together. An attacker hoping for acquiescence may encounter collective resistance even with correct equilibrium beliefs. The mixed equilibrium captures this uncertainty in coordination. Russia's 2022 invasion of Ukraine is a case in point. While Russia may have expected a \DC\ response in which the invasion would be interpreted as a particular dispute involving Ukraine, the response of Ukraine and the Western alliance reflected a competing interpretation: that the attack revealed a broader challenge rather than an isolated conflict \citep{Zelenskyy2023NDU}.

\begin{figure}[htbp]
\centering
\includegraphics[width=.94\textwidth]{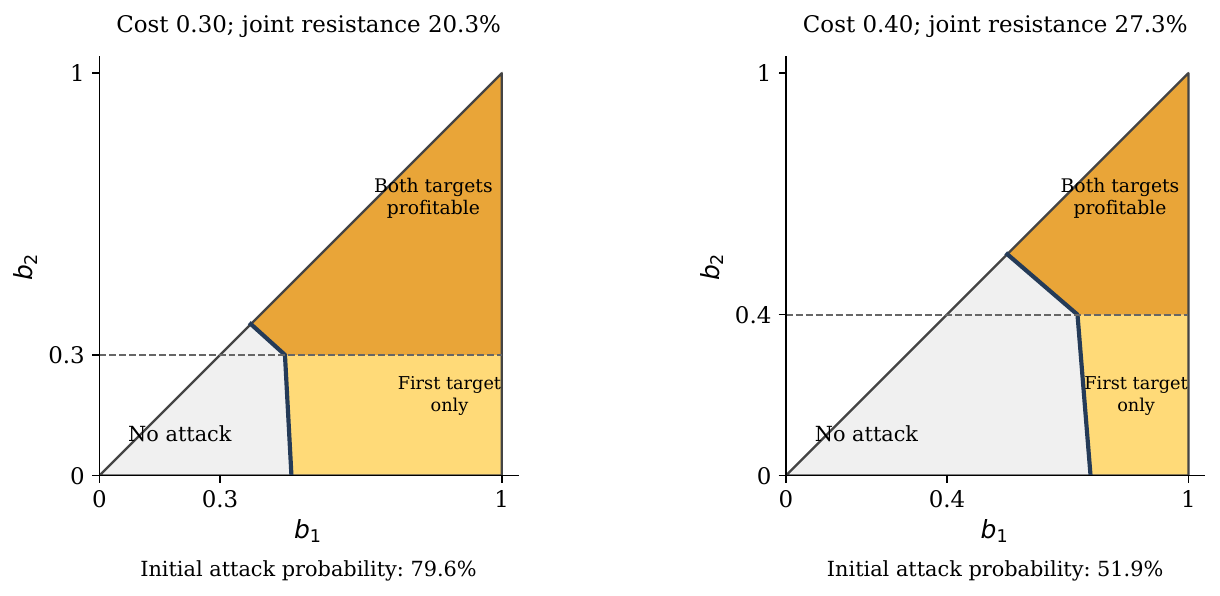}
\caption{Mixed-equilibrium attack pools under uniform benefits, with \(p_1=.8\), \(p_2=.2\), and \(d/\ell=.4\). The panel headings report joint resistance conditional on an initial attack. The two colored regions together constitute the initial attack set. Conditional on no joint resistance, the light region leads to an attack confined to \(\Vone\), and the dark region to a sequential attack on \(\Vtwo\) as well. The dark region remains \(\kappa=.6\) of the attack pool in both panels. A joint fight, when it occurs, involves both victims regardless of the type's color.}
\label{fig:mixed-partition}
\end{figure}

\paragraph{Two selection benchmarks.} These considerations motivate reporting comparative statics under two equilibrium selections. The \emph{attacker-preferred} selection picks the \DC\ equilibrium whenever it exists. It highlights what the attacker can get away with: surrender maximizes every attacker type's payoff and gives the largest initial attack set. The \emph{mixed} selection picks a mixed equilibrium whenever one exists, and otherwise coincides with the attacker-preferred selection. It offers a somewhat more neutral benchmark, capturing uncertainty about whether collective resistance will materialize. If there are several mixed equilibria, we take the first one reached as the probability of resistance is raised from zero; \Cref{oa:mixed} specifies the selection. When \(\phi(c)>\kappa\), so that \DC\ is not sustainable, \Cref{cor:unique-CR} makes the benchmarks coincide. The extensions in \Cref{sec:extensions} retain attacker-preferred selection.

\subsection{Comparative static I: the cost of attack}\label{sec:cost}

The first comparative static concerns the cost of attack. A higher cost mechanically makes the attack less profitable. Conditional on observing an attack, however, a higher cost can lead the bystander to infer that the attacker's motive is more specific to \(\Vone\). This latter feature can produce a striking non-monotonicity: as \(c\) rises, \(\Vtwo\) may become less willing to join the resistance, and the attacker may be more likely to employ divide-and-conquer attacks.

To obtain a clean analysis, we impose the following condition on the marginal distributions of \(b_1\) and \(b_2\).

\begin{assumption}[Relative survival of target benefits]\label{ass:LR}
For \(0\le c<c'<1\),
\[
\frac{1-F_1(c')}{1-F_1(c)}
\ge
\frac{1-F_2(c')}{1-F_2(c)}.
\]
\end{assumption}

In words, when the cost rises from \(c\) to \(c'\), the first target remains profitable for a weakly larger share of the types for which it was profitable than the second target does. Equivalently, the first target's marginal hazard rate is no greater than the second's. This condition is weaker than likelihood-ratio domination, and it is exactly what the posterior comparison requires: since \(b_2\ge c\) implies \(b_1\ge c\) on \(B\), \(\phi(c)=[1-F_2(c)]/[1-F_1(c)]\), so the inequality says that \(\phi(c')\le\phi(c)\).

Under the ordered-independent-draw condition in \Cref{def:ordered}, \Cref{ass:association,ass:LR} hold and
\[
\phi(c)=\frac{1-F(c)}{1+F(c)}.
\]
(\Cref{oa:orderstat,oa:mixed} verify these claims.) This posterior strictly decreases from one to zero.

\begin{lemma}[A costly first attack looks less systemic]\label{lem:LRphi}
Under \Cref{ass:LR}, \(\phi(c)\) is weakly decreasing and \(\phi(0)=1\).
\end{lemma}

For the switch from collective resistance to \DC\ to occur at a unique interior cost, we additionally require that \(\phi\) is strictly decreasing and \(\lim_{c\uparrow1}\phi(c)<\kappa\). These requirements hold under \Cref{def:ordered}.

\begin{proposition}[Cost and the two selection benchmarks]\label{prop:cost}
Suppose \Cref{ass:victims,ass:association,ass:LR} hold, \(\phi\) is strictly decreasing, and \(\lim_{c\uparrow1}\phi(c)<\kappa\). Let \(\hat c\) be the unique cost with \(\phi(\hat c)=\kappa\).
\begin{enumerate}[label=(\roman*),leftmargin=2.7em]
\item Below \(\hat c\), both benchmarks select collective resistance. Attack probability weakly decreases with cost, reaching zero at \(c=p_2\) if \(p_2<\hat c\).
\item At \(\hat c\), attack probability jumps upward under both benchmarks, from a left limit of \(\mu(J(\hat c))\) to a right limit of \(1-F_1(\hat c)\). The left limit is zero exactly when \(\hat c\ge p_2\).
\item Above \(\hat c\), attack probability decreases under attacker-preferred selection. It also decreases along smooth segments of the mixed benchmark where both victims randomize.
\end{enumerate}
\end{proposition}

The central result is the upward jump. A higher cost first discourages attacks against a united defense. At \(\hat c\), however, it becomes possible for the bystander to interpret an attack as localized and to stay out of the first victim's fight. Victims surrender under attacker-preferred selection, while collective resistance becomes arbitrarily unlikely just above the threshold under the mixed benchmark. Both changes admit attackers who would have been deterred by joint defense.

The same switch can also raise the second victim's own exposure. Under collective resistance, \(\Vtwo\) is drawn into conflict only for types in \(J(c)\); under \DC, she is attacked whenever \(b_2\ge c\). If \(p_2<2\hat c/(1+\hat c)\), then \(J(\hat c)\subseteq\{b_2\ge\hat c\}\), with strict inclusion in probability, so the probability that \(\Vtwo\) is attacked jumps upward at \(\hat c\), from \(\mu(J(\hat c))\) to \(\Pr(b_2\ge\hat c)\), under both benchmarks. In \Cref{fig:mixed-cost}, it jumps from zero to \(.5625\).

Above the threshold, higher cost again discourages attacks under pure \DC. Under mixing, resistance adjusts to keep the bystander at her belief threshold. The argument in \Cref{oa:mixed-statics} shows why attacks also decline when both victims mix, even as resistance adjusts. The distinctive implication of mixing concerns the second victim. Conditional on an initial attack and no joint resistance, the chance of a subsequent attack on \(\Vtwo\) stays at \(\kappa\). Under pure \DC, that chance is \(\phi(c)\): it falls with cost and, under \Cref{def:ordered}, tends to zero.

\begin{figure}[htbp]
\centering
\includegraphics[width=\textwidth]{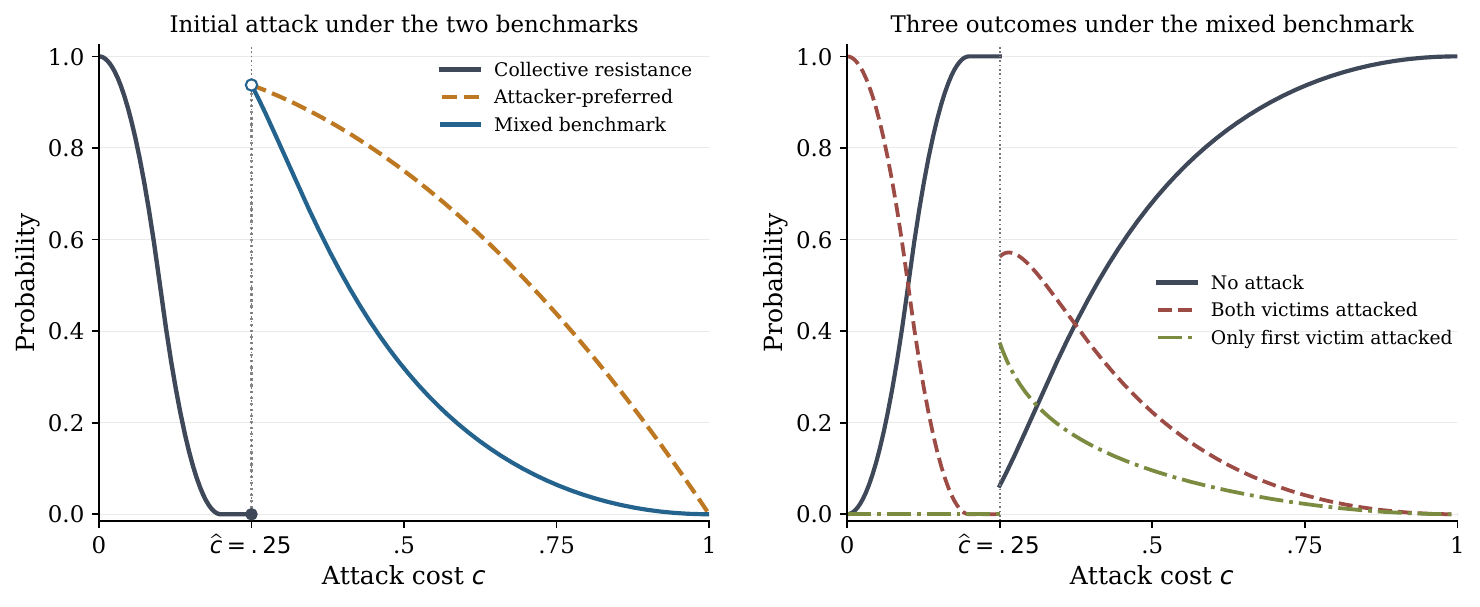}
\caption{Cost and the two selection benchmarks. Benefits are uniform, \(p_1=.8\), \(p_2=.2\), \(d/\ell=.4\), and \(\kappa=.6\). Below \(\hat c=.25\), both benchmarks select the unique collective-resistance outcome; attacks are fully deterred for \(.2\le c<.25\). At \(c=.25\), attacks jump from zero to a right limit of \(.9375\) under both benchmarks. Above it, the mixed benchmark has fewer attacks than pure \DC. The right panel separates no attack, an attack confined to the first victim, and an attack involving both victims; these probabilities sum to one.}
\label{fig:mixed-cost}
\end{figure}

Interpreting \(c\) as moral or reputational cost gives the result a historical reading. When aggression is viewed as norm-breaking or ``unthinkable,'' a first attack may be rationalized as exceptional rather than systemic. The remilitarization of the Rhineland in 1936 and the annexation of Crimea in 2014 were both framed by the aggressors as particular claims rather than as transparent announcements of unconstrained campaigns \citep{USHMMRhineland,UNGA68262}. The model does not require bystanders to be correct in accepting such frames; it shows how a costly first move can make a localized interpretation more plausible.

\subsection{Comparative static II: correlation}\label{sec:correlation}

The second comparative static concerns how connected the victims' fates appear to be. If the attacker's benefits from the two victims are highly correlated, an attack on one is strong evidence that the other is also at risk. If the benefits are weakly correlated, the same attack is easier to interpret as particularized.

To express this cleanly, we use the order-statistic microfoundation but allow the primitive values \((X_1,X_2)\) to be dependent. They have a common marginal distribution \(F\) and a copula \(C_\rho\):
\[
  \Pbb(X_1\le x_1,X_2\le x_2)=C_\rho(F(x_1),F(x_2)).
\]
The copula separates marginal attractiveness from dependence. The parameter \(\rho\) measures how tightly the two target values move together while holding each marginal fixed. Higher \(\rho\) means that if one target is valuable to the attacker, the other is more likely to be valuable as well. Let \(b_1=\max\{X_1,X_2\}\) and \(b_2=\min\{X_1,X_2\}\). Since \(b_2\ge c\) means that both primitive values exceed \(c\), while \(b_1\ge c\) means that at least one does, the posterior is
\begin{equation}\label{eq:copulaphi}
  \phi_\rho(c)=\frac{1-2F(c)+C_\rho(F(c),F(c))}{1-C_\rho(F(c),F(c))}.
\end{equation}

\begin{proposition}[Correlation and the two selection benchmarks]\label{prop:correlation}
Suppose \(F\) is fixed and \(C_\rho(u,u)\) is increasing in \(\rho\) for every \(u\in(0,1)\).
\begin{enumerate}[label=(\roman*),leftmargin=2.7em]
\item Higher dependence weakly raises \(\phi_\rho(c)\) and makes \DC\ harder to sustain. Where attacker-preferred selection uses \DC, it weakly lowers initial attack probability.
\item Under mixed selection, higher dependence can instead raise initial attack probability, even with uniform underlying benefits and an affiliated joint density.
\end{enumerate}
\end{proposition}

Under the copula condition, higher dependence raises \(\phi_\rho(c)\), the posterior the bystander forms when \(\Vone\) is expected to surrender, and shrinks the region supporting \DC. The condition holds in standard families ordered by concordance \citep{Nelsen2006,Joe2014}. Under the mixed benchmark, the effect on attacks can be reversed: because higher dependence makes an attack more alarming, less resistance is needed to keep the bystander's posterior at \(\kappa\), and with less resistance more types can attack. \Cref{oa:correlation-mixed} verifies a smooth example with opposite attack responses under the two benchmarks.

The parameter \(\rho\) captures perceived linkage among victims' fates. When \(\rho\) is low, an attack on one victim is easier to interpret as target-specific. When \(\rho\) is high, the same attack is stronger evidence of a broader campaign. This is the Niemöller inference: collective resistance is easier to sustain when victims understand their risks as linked.

R\'akosi's ``salami tactics'' in Hungary illustrate the low-correlation inference. Each faction could view the attack on another faction as a separate political struggle, rather than as evidence that the whole opposition would eventually be eliminated \citep{TimeRakosi1947,TNSRSalami2021}. British interventions among weakly coordinated Indian states suggest a related interpretation: each polity faced incentives to treat British intervention elsewhere as a localized dispute rather than as evidence of a common imperial project \citep{BritannicaPrincelyStates,BritannicaEIC}. Counter-propaganda works in the opposite direction by raising perceived linkage. An attack on Ukraine, for example, may be framed not as a special dispute over Ukraine but as evidence of broader imperial ambition  \citep{Zelenskyy2023NDU}.

\subsection{Comparative static III: asymmetry and specialness}\label{sec:asymmetry}

The third comparative static concerns the publicly understood specialness of \(\Vone\). Let \(a\ge0\) be a commonly known additional benefit from conquering \(\Vone\), while the benefit from \(\Vtwo\) remains \(b_2\). This can represent a territorial claim or strategic asset specific to the first target. When \(\Vone\) is expected to surrender, the first attack occurs when \(b_1+a\ge c\). The posterior after an attack on \(\Vone\) is
\begin{equation}\label{eq:phia}
  \phi_a(c):=\Pbb(b_2\ge c\mid b_1+a\ge c).
\end{equation}
Because \(b_2\ge c\) implies \(b_1\ge c\) on \(B\),
\[
  \phi_a(c)=\frac{\Pbb(b_2\ge c)}{\Pbb(b_1\ge c-a)},
\]
where \(\Pbb(b_1\ge c-a)=1\) when \(c-a\le0\).

\begin{proposition}[Specialness and the two selection benchmarks]\label{prop:asymmetry}
As the additional first-target benefit \(a\) rises:
\begin{enumerate}[label=(\roman*),leftmargin=2.7em]
\item \(\phi_a(c)\) weakly decreases and the region supporting \DC\ expands. Where attacker-preferred selection uses \DC, initial attack probability weakly increases.
\item Along smooth segments of the selected mixed benchmark where both victims randomize, initial attack probability also weakly increases.
\end{enumerate}
\end{proposition}

A publicly understood reason for valuing the first target more highly admits attackers with less systemic motives. Under pure \DC, this dilutes the bystander's threat inference. Under both-victim mixing, resistance adjusts to preserve the belief threshold while the attack pool weakly expands. The argument is in \Cref{oa:specialness-mixed}. Particularizing rhetoric may make such stakes salient, as claims about ethnic Germans did in the Sudetenland crisis \citep{USHMMMunich}; information and framing are studied below.

\section{Richer Theories of Divide-and-Conquer}\label{sec:extensions}

The baseline model fixes a simple informational environment:  the victims are Bayesian, and no prior institution changes the cost of joining or staying out. This section considers four extensions. First, behavioral responses
change how the bystander acts on the danger posterior. Second, rhetoric and propaganda
change the information on which the bystander conditions. Third, treaty commitments alter
both payoffs and inference. Fourth, with more than two victims, downstream protection can
make an intermediate victim less willing to defend the first one. Each extension changes one
component of the environment and records the corresponding equilibrium implication; the
proofs, in \Cref{app:extension-proofs}, collect the algebra behind the derived cutoffs. For simplicity, we focus throughout this section on the attacker-preferred equilibrium: \DC\ is selected whenever it is sustainable.

\subsection{Behavioral responses}
\label{sec:behavioral}

A Bayesian bystander stays out after an attack on \(\Vone\) if and only if
\(\phi(c)\le\kappa\).  Behavioral responses can affect this comparison in two ways.  The
first changes the belief the bystander allows herself to hold.  The second changes the
perceived consequence of staying out.  We use stripped-down formulations because the goal
is not to build a separate theory of motivated cognition, but to see how such cognition
enters the same posterior-threshold logic.

Our reduced-form formulation of \emph{wishful thinking} is inspired by the idea of beliefs as consumption goods in
\citet{BrunnermeierParker2005}.  Believing that one is safe is comforting; believing that
one is next is painful.  After observing the attack, \(\Vtwo\) may choose a subjective
posterior \(m\in[\zeta\phi(c),\phi(c)]\), where \(\zeta\in[0,1]\).  The parameter \(\zeta\)
measures discipline: \(\zeta=1\) gives Bayesian beliefs, while a lower \(\zeta\) permits more
optimism.  If \(r\in\{S,J\}\) denotes staying out or joining, \(\Vtwo\)'s loss criterion is
\(L(r;\phi(c))+\theta L(r;m)\), where \(L(S;q):=q\ell\) and \(L(J;q):=\kappa\ell\) are the expected losses from staying out and from joining under belief \(q\), and \(\theta\ge0\) is the weight on anticipatory utility.
Throughout this subsection, when a behavioral cutoff is invoked, we focus on the interior
case in which the cutoff exists.

\begin{proposition}[Wishful thinking]
\label{prop:wishful-thinking}
Assume \Cref{ass:LR}, \(\phi\) strictly decreasing, \(\theta>0\), and \(\zeta<1\). Under wishful thinking, there exists \(\hat c^{WT}<\hatc\) such that \DC\ is sustained if and only if \(c\ge \hat c^{WT}\).
\end{proposition}

Wishful thinking expands the D\&C region by making the bystander less willing to act on the Bayesian posterior. The bystander understands that staying out can be materially dangerous, but she also receives an anticipatory benefit from entertaining a reassuring belief. Since the most reassuring feasible belief is \(m=\zeta\phi(c)\), the subjective component of the decision problem tilts the comparison toward inaction. Formally, staying out is optimal if and only if \(\phi(c)\le(1+\theta)\kappa/(1+\theta\zeta)\), and this effective belief threshold exceeds \(\kappa\) whenever \(\theta>0\) and \(\zeta<1\).

The result creates an intermediate region \([\hat c^{WT},\hatc)\). In this region, the attack is objectively alarming enough that a Bayesian bystander would join, but not so alarming that wishful thinking is disciplined away. The bystander can still sustain the thought that the first attack is localized. This is precisely the psychological margin from which divide-and-conquer benefits: the first victim is exposed not because the second victim is indifferent, but because the second victim finds the systemic interpretation too painful to accept.

Under \emph{magical thinking}, the distortion is not a direct change in the posterior over
types.  It is a perceived causal effect of the bystander's own action.  The bystander
imagines that staying out places her in the attacker's good graces.  In the material game,
neutrality does not reduce \(\A\)'s benefit from conquering \(\Vtwo\).  A magical thinker
acts as if it does.  We represent this by a credit \(\omega\ge0\): if \(\Vtwo\) stays out, she
perceives the future-attack condition to be \(b_2\ge c+\omega\) rather than \(b_2\ge c\).

\begin{proposition}[Magical thinking]
\label{prop:magical-thinking}
Assume \Cref{ass:LR} and \(\phi\) strictly decreasing. Under magical thinking with \(\omega>0\), \DC\ is sustained whenever it is sustained in the Bayesian benchmark. Moreover, there exists \(\hat c^{MT}_\omega<\hatc\) such that \DC\ is sustained for every \(c\ge\hat c^{MT}_\omega\).
\end{proposition}

Magical thinking operates through a different channel. The bystander does not merely shade the probability of danger; she misperceives the causal effect of her own neutrality. The perceived credit \(\omega\) raises the threshold for a later attack from \(c\) to \(c+\omega\). Under the maintained positive-density assumption, this strictly lowers the perceived posterior at the Bayesian cutoff, so the range of costs supporting D\&C extends strictly below \(\hatc\).

Wishful thinking and magical thinking affect different margins. The former changes the belief the bystander permits herself to hold; the latter changes the perceived consequence of staying out. Counter-propaganda is therefore directed at both margins: it makes the comforting posterior harder to maintain and challenges the causal claim that neutrality buys safety.

Munich is a useful but imperfect historical illustration of the first margin. Chamberlain believed that Britain could negotiate with Hitler in good faith and hoped that satisfying some Nazi demands---most visibly by conceding the Sudetenland in exchange for a promise of no further territorial claims---would prevent a wider war \citep{USHMMChamberlain,NationalArchivesMunich}. This made it possible to treat the crisis as a limited grievance rather than as a signal of broader ambition.\footnote{This should not be read as a pure case of psychological naivet\'e: British appeasement also reflected military, domestic, imperial, and diplomatic constraints, and \citet{RipsmanLevy2008} argue that it was largely a strategy of buying time. The behavioral point is narrower. Wishful thinking captures the force that makes a limited-grievance interpretation attractive when recognizing the systemic interpretation is painful.}

Neutrality provides a closer analogy for magical thinking. Dutch neutrality had been the country's foreign policy for a century and had kept the Netherlands out of World War I; Germany invaded nonetheless in May 1940 \citep{AnneFrankNetherlands}. Belgium similarly adopted an independent foreign policy in 1936 in the hope of avoiding war, but Germany invaded Belgium on May 10, 1940 \citep{BelgiumMonarchy}. The point is not that neutrality was irrational. It reflected past experience and strategic constraints. But these cases illustrate the same perceived causal margin---the hope that staying out of a conflict will itself preserve safety, even when the aggressor's material incentive to attack is not thereby removed.

\subsection{Rhetoric and propaganda}\label{sec:rhetoric}

Rhetoric can work through two channels: it may convey information about the attacker's type, or it may make a particular interpretation psychologically salient. The formal result studies the first channel; we return to the second after the proposition. The information lies not in the facts that rhetoric invokes, which may be common knowledge, but in the attacker's choice to invoke them. A historical claim, a disputed border, or the status of an ethnic minority may be known to everyone; that the attacker dwells on such a grievance reveals how much it matters to her, and hence how large her stake in \(\Vone\) is relative to her stake in \(\Vtwo\). We model such rhetoric as a public signal, observed after the initial attack and before resistance decisions, that reveals which of two cells of the type space contains the attacker. As in the voluntary-disclosure literature \citep{Grossman1981,Milgrom1981,MilgromRoberts1986}, the signal is taken to be truthful; it may also come from a source other than the attacker, such as independent intelligence.

Within \(B\), fix \(\Delta>0\), and let
\[
  L_\Delta:=\{(b_1,b_2)\in B:b_1-b_2>\Delta\},
  \qquad
  H_\Delta:=\{(b_1,b_2)\in B:b_1-b_2\le\Delta\}.
\]
A particularizing signal reveals \(L_\Delta\); it says that \(\A\) has a substantially stronger reason to conquer \(\Vone\) than \(\Vtwo\). A systemic signal reveals \(H_\Delta\); it says that the two target values are close. A sufficiently particularizing cell can support staying out, while a sufficiently systemic cell can support collective resistance; \Cref{prop:propaganda-formal} gives the relevant ranges of \(\Delta\) for propaganda and for counter-propaganda.\footnote{The construction has a fixed-point character because the posterior in a signal cell depends on the continuation assigned to that cell. The verification is nevertheless well defined. Once a cell is assigned to \DC, the relevant attack pool is \(\{b_1\ge c\}\) intersected with that cell; once a cell is assigned to collective resistance, the relevant attack pool is \(J(c)\) intersected with that cell. The proof checks the corresponding posterior condition for each cell.}

\begin{proposition}[Propaganda and counter-propaganda]\label{prop:propaganda-formal}
Assume \Cref{ass:positive-selection}.
\begin{enumerate}[label=(\roman*)]
\item If \(c<\hatc\), there exists \(\bar\Delta(c)\in(0,1-c)\) such that, for every \(\Delta\in(\bar\Delta(c),1-c)\), a particularizing signal induces \DC\ on \(L_\Delta\) and collective resistance on \(H_\Delta\).
\item If \(c>\hatc\), there exists \(\underline\Delta(c)>0\) such that, for every \(\Delta\in(0,\underline\Delta(c))\), a systemic signal induces collective resistance on \(H_\Delta\) and leaves \DC\ on \(L_\Delta\).
\end{enumerate}
\end{proposition}

For any \(\Delta\) in the relevant interval, the post-signal attack set within the support \(B\) is
\[
  (L_\Delta\cap\{b_1\ge c\})\cup(H_\Delta\cap J(c)).
\]
This expression has different interpretations in the two cases. Propaganda adds a D\&C region below the line \(b_2=b_1-\Delta\): in that region, all types with \(b_1\ge c\) attack \(\Vone\). Counter-propaganda instead forces types above the line to face collective resistance: in that region, only those in \(J(c)\) still attack. Thus rhetoric does not simply switch the whole type space from one equilibrium to another. It changes the continuation on the part of the type space to which the signal directs attention.

\Cref{fig:propaganda-before-after} illustrates the first case. Without the signal, since \(c<\hatc\), the selected equilibrium is collective resistance and only types in \(J(c)\) attack. After a sufficiently strong particularizing signal, types below \(b_2=b_1-\Delta\) face D\&C. Types above the line still face collective resistance; among them, only those above the \(-45^\circ\) boundary \(p_2(b_1+b_2)=2c\) attack. The signal therefore expands the attack set by carving out a new D\&C region, while leaving the systemic region disciplined by collective resistance.

\begin{figure}[htb]
\centering
\begin{minipage}{0.49\textwidth}
\centering
\begin{tikzpicture}[scale=5.65, font={\fontsize{11pt}{13pt}\selectfont}]
  \def\c{0.24}
  \def\a{0.80}
  \def\yr{0.60}
  \draw[->] (0,0)--(1.08,0) node[right] {$b_1$};
  \draw[->] (0,0)--(0,1.08) node[above] {$b_2$};
  \draw[densely dotted] (0,0)--(1,1);
  \draw[dashed] (\c,0)--(\c,\c)--(1,\c);
  \begin{scope}
    \clip (0,0)--(1,0)--(1,1)--cycle;
    \fill[colCR!50] (\a,\a)--(1,\yr)--(1,1)--cycle;
  \end{scope}
  \draw[very thick,colCR!90!black] (\a,\a)--(1,\yr);
  \draw (0,0)--(1,1)--(1,0);
  \node[below] at (\c,0) {$c<\hat c$};
  \node[left] at (0,\c) {$c$};
  \node[black] at (.92,.82) {CR};
  \node at (.5,-.17) {Before propaganda};
\end{tikzpicture}
\end{minipage}\hfill
\begin{minipage}{0.49\textwidth}
\centering
\begin{tikzpicture}[scale=5.65, font={\fontsize{11pt}{13pt}\selectfont}]
  \def\c{0.24}
  \def\D{0.18}
  \def\a{0.80}
  \def\xi{0.89}
  \def\yi{0.71}
  \begin{scope}
    \clip (0,0)--(1,0)--(1,1)--cycle;
    \fill[colDCone!50] (\c,0)--(1,0)--(1,1-\D)--(\c,\c-\D)--cycle;
    \fill[colDCboth!55] (\c+\D,\c)--(1,\c)--(1,1-\D)--cycle;
  \end{scope}
  \begin{scope}
    \clip (0,0)--(1,0)--(1,1)--cycle;
    \fill[colCR!50] (\a,\a)--(\xi,\yi)--(1,1-\D)--(1,1)--cycle;
  \end{scope}
  
  \draw[->] (0,0)--(1.08,0) node[right] {$b_1$};
  \draw[->] (0,0)--(0,1.08) node[above] {$b_2$};
  
  \draw[very thick,dcpurple] (\D,0)--(1,1-\D);
  \draw[very thick,colCR!90!black] (\a,\a)--(\xi,\yi);
  \draw (0,0)--(1,1)--(1,0);
  \node[rotate=45] at (.70,.42) {$b_2=b_1-\Delta$};
  \draw[dashed] (\c,0)--(\c,\c)--(1,\c);
  \node[black,align=center] at (.75,.11) {\DC};
  \node[black,align=center] at (.92,.82) {CR};
  \node[below] at (\c,0) {$c$};
  \node[left] at (0,\c) {$c$};
  \node at (.5,-.17) {After propaganda};
\end{tikzpicture}
\end{minipage}
\caption{Propaganda when \(c<\hat c\). Before propaganda, only the collective-resistance attack set \(J(c)\) attacks. After a sufficiently strong particularizing signal, types below the line \(b_2=b_1-\Delta\) face \DC\ (light and dark yellow combined), while types above the line face collective resistance. Among the latter, only those above the \(-45^\circ\) boundary \(p_2(b_1+b_2)=2c\) still attack (red).}
\label{fig:propaganda-before-after}
\end{figure}

\Cref{fig:counter-before-after} illustrates the second case.  For \(c>\hatc\), without the systemic signal, D\&C prevails throughout \(\{b_1\ge c\}\cap B\). After the signal, the region above \(b_2=b_1-\Delta\) is reclassified as systemic and faces collective resistance. The attack set shrinks because types in that region attack only if they are willing to fight a two-front war. Below the line, the particularizing interpretation remains strong enough for D\&C to continue.

\begin{figure}[htb]
\centering
\begin{minipage}{0.49\textwidth}
\centering
\begin{tikzpicture}[scale=5.65, font={\fontsize{11pt}{13pt}\selectfont}]
  \def\c{0.40}

  \begin{scope}
    \clip (0,0)--(1,0)--(1,1)--cycle;
    \fill[colDCone!50] (\c,0) rectangle (1,1);
    \fill[colDCboth!55] (\c,\c) rectangle (1,1);
  \end{scope}

  \draw[->] (0,0)--(1.08,0) node[right] {\(b_1\)};
  \draw[->] (0,0)--(0,1.08) node[above] {\(b_2\)};

  \draw[dashed] (\c,0)--(\c,\c)--(1,\c);
  \draw (0,0)--(1,1)--(1,0);

  \node[below] at (\c,0) {\(c>\hat c\)};
  \node[left] at (0,\c) {\(c\)};
  \node[black,align=center] at (.75,.50) {\DC\ (both)};
  \node[black,align=center] at (.75,.20) {\DC\ (\(\Vone\) only)};
  \node at (.5,-.17) {Before counter-propaganda};
\end{tikzpicture}
\end{minipage}\hfill
\begin{minipage}{0.49\textwidth}
\centering
\begin{tikzpicture}[scale=5.65, font={\fontsize{11pt}{13pt}\selectfont}]
  \def\c{0.40}
  \def\D{0.26}
  \def\a{0.76}
  \def\xi{0.89}
  \def\yi{0.63}

  \begin{scope}
    \clip (0,0)--(1,0)--(1,1)--cycle;
    \fill[colDCone!50] (\c,0)--(1,0)--(1,1-\D)--(\c,\c-\D)--cycle;
    \fill[colDCboth!55] (\c+\D,\c)--(1,\c)--(1,1-\D)--cycle;
  \end{scope}
  \begin{scope}
    \clip (0,0)--(1,0)--(1,1)--cycle;
    \fill[colCR!50] (\a,\a)--(\xi,\yi)--(1,1-\D)--(1,1)--cycle;
  \end{scope}

  \draw[->] (0,0)--(1.08,0) node[right] {\(b_1\)};
  \draw[->] (0,0)--(0,1.08) node[above] {\(b_2\)};

  \draw[very thick,dcpurple] (\D,0)--(1,1-\D);
  \draw[very thick,colCR!90!black] (\a,\a)--(\xi,\yi);
  \draw[dashed] (\c,0)--(\c,\c)--(1,\c);
  
  \draw (0,0)--(1,1)--(1,0);

  \node[rotate=45] at (.675,.545) {\(b_2=b_1-\Delta\)};
  \node[black,align=center] at (.75,.20) {\DC};
  \node[black,align=center] at (.90,.78) {CR};
  \node[below] at (\c,0) {\(c\)};
  \node[left] at (0,\c) {\(c\)};
  \node at (.5,-.17) {After counter-propaganda};
\end{tikzpicture}
\end{minipage}
\caption{Counter-propaganda when \(c>\hat c\). Before the counter-message, \DC\ prevails throughout \(\{b_1\ge c\}\cap B\). After a sufficiently sharp systemic signal, types above the line \(b_2=b_1-\Delta\) face collective resistance; among them, only those above the \(-45^\circ\) boundary \(p_2(b_1+b_2)=2c\) still attack (red). Types below the line continue to face \DC\ (light and dark yellow combined).}
\label{fig:counter-before-after}
\end{figure}

The historical interpretation is direct. Propaganda tries to particularize the first victim: ``this case is exceptional; the grievance stops here.'' Counter-propaganda seeks to situate the first attack within a broader pattern. Nazi claims about the special status of ethnic Germans in neighboring territories and Putin's insistence that Ukraine occupies a unique place in Russian history are examples of particularizing frames \citep{USHMMMunich,Putin2021}. The opposing rhetorical strategy is to insist that the first attack reveals a larger project. Zelenskyy's warning to European audiences that Russian weapons aimed at Ukraine are already pointed at Ukraine's neighbors is an example of this systemic counter-frame \citep{Zelenskyy2023Munich}. In the language of the model, the first type of message pushes the bystander toward \(L_\Delta\), while the second pushes her toward \(H_\Delta\).

Even when rhetoric carries no such information, it may still work through the behavioral mechanisms in \Cref{sec:behavioral}. A wishful bystander wants a reason to believe that the attack is localized; a particularizing story supplies that reason even if it would not persuade a detached Bayesian observer. Formally, the story relaxes the effective discipline \(\zeta\) on the subjective posterior. Magical thinking is also available: the attacker may imply that neutrality will be remembered or rewarded, inducing the bystander to act as if staying out creates the credit \(\omega\). Counter-propaganda attacks both channels by making the comforting belief harder to hold and by denying that silence changes the attacker's future incentives.

\subsection{Treaty making}\label{sec:treaties}

Treaties affect divide-and-conquer because they change both material incentives and
inference.  A treaty may lower the bystander's perceived risk of being attacked later, but
it may also change the cost of helping the first victim now.  We consider two polar cases.
A neutrality or non-aggression treaty between \(\A\) and \(\Vtwo\) raises \(\A\)'s cost of
attacking \(\Vtwo\) from \(c\) to \(c+\gamma\), where \(\gamma>0\) is the penalty for breaching the treaty, and makes \(\Vtwo\)'s participation in a
fight against \(\A\) cost \(p_2\ell+d+\gamma\).  A collective-defense treaty between \(\Vone\)
and \(\Vtwo\) instead imposes \(\gamma\) on a victim who refuses to come to the other's
defense.

The first treaty is protective in a narrow bilateral sense: it makes a later attack on
\(\Vtwo\) more costly.  But that protection is itself a source of reassurance.  After an
attack on \(\Vone\), \(\Vtwo\) is a later target not whenever \(b_2\ge c\), but only in the more demanding event \(b_2\ge c+\gamma\).  If joining \(\Vone\)'s fight also violates the treaty, both effects
move in the same direction.

\begin{proposition}[Neutrality treaties]\label{prop:treaty-A2}
Fix \(c<\hatc\). If \(0<\gamma<1-c\) and \(\Pr(b_2\ge c+\gamma\mid b_1\ge c)\ell\le p_2\ell+d+\gamma\), then a treaty between \(\A\) and \(\Vtwo\) induces a \DC\ equilibrium under attacker-preferred selection.
\end{proposition}

The treaty has two reinforcing effects. First, because \(\A\) must pay the additional cost \(\gamma\) to attack \(\Vtwo\), the event that \(\Vtwo\) is a future target becomes \(b_2\ge c+\gamma\) rather than \(b_2\ge c\). The posterior risk after an attack on \(\Vone\) therefore falls. Second, because \(\Vtwo\) would break the treaty by joining \(\Vone\)'s fight, the cost of joining rises. The displayed inequality is exactly the condition under which these two effects make staying out optimal for \(\Vtwo\).

Relative to the no-treaty collective-resistance equilibrium, the probability of an attack on \(\Vone\) rises. Without the treaty and with \(c<\hatc\), the selected attack set is \(J(c)\); with the treaty-induced D\&C equilibrium, all types with \(b_1\ge c\) attack \(\Vone\). Observe that \(J(c)\subseteq\{b_1\ge c\}\), with strict containment under the maintained density conditions. Moreover, for suitable \(\gamma\) and sufficiently small \(p_2\), the treaty can also increase the probability that \(\Vtwo\) will eventually be attacked. In this case, a treaty designed to prevent aggression increases the chance of aggression. Joint resistance makes the no-treaty attack set very small when \(p_2\) is small, while the treaty-induced future-attack set \(\{b_2\ge c+\gamma\}\) retains positive probability.

\begin{figure}[htb]
\centering
\begin{minipage}{0.49\textwidth}
\centering
\begin{tikzpicture}[scale=5.75, font={\fontsize{11pt}{13pt}\selectfont}]
  \def\c{0.30}
  \def\S{1.56}
  \def\a{0.78}
  \def\yr{0.56}

  \begin{scope}
    \clip (0,0)--(1,0)--(1,1)--cycle;
    \fill[colHypL] (\c,0) rectangle (1,\c);
    \fill[colHypD!60] (\c,\c)--(1,\c)--(1,1)--cycle;
  \end{scope}

  \begin{scope}
    \clip (0,0)--(1,0)--(1,1)--cycle;
    \fill[colCR!50] (\a,\a)--(1,\yr)--(1,1)--cycle;
  \end{scope}

  \draw[->] (0,0)--(1.08,0) node[right] {$b_1$};
  \draw[->] (0,0)--(0,1.08) node[above] {$b_2$};
  
  \draw[very thick,colCR!90!black] (\a,\a)--(1,\yr);

  \node[below] at (\c,0) {$c<\hat c$};
  \node[left] at (0,\c) {$c$};
  \draw[dashed] (\c,0)--(\c,\c)--(1,\c);
  \draw (0,0)--(1,1)--(1,0);
  \node[black,align=center] at (.69,.15) {would attack\ only $\Vone$};
  \node[black,align=center] at (.78,.43) {would attack\\ both};
  \node[black,align=center] at (.90,.80) {CR};
  \node at (.5,-.17) {Without treaty};
\end{tikzpicture}
\end{minipage}\hfill
\begin{minipage}{0.49\textwidth}
\centering
\begin{tikzpicture}[scale=5.75, font={\fontsize{11pt}{13pt}\selectfont}]
  \def\c{0.30}
  \def\ct{0.60}
  
  \begin{scope}
    \clip (0,0)--(1,0)--(1,1)--cycle;
    \fill[colDCone!50] (\c,0) rectangle (1,\ct);
    \fill[colDCboth!55] (\ct,\ct)--(1,\ct)--(1,1)--cycle;
  \end{scope}

  \draw[->] (0,0)--(1.08,0) node[right] {$b_1$};
  \draw[->] (0,0)--(0,1.08) node[above] {$b_2$};

  \node[below] at (\c,0) {$c$};
  \node[left] at (0,\c) {$c$};
  \node[left] at (0,\ct) {$c+\gamma$};
  \draw[dashed] (\c,0)--(\c,\c); 
  \draw[dashed] (0,\ct)--(1,\ct);
  \draw (0,0)--(1,1)--(1,0);
  \node[black,align=center] at (.73,.28) {\DC\ ($\Vone$ only)};
  \node[black,align=center] at (.88,.70) {\DC\\ (both)};
  \node at (.5,-.17) {With treaty $\A$--$\Vtwo$};
\end{tikzpicture}
\end{minipage}
\caption{A treaty between \(\A\) and \(\Vtwo\). Left: without treaty, the gray regions show the potential \DC\ decomposition, while the red wedge is the actual collective-resistance attack set. Right: with treaty, the cutoff for attacking \(\Vtwo\) rises from \(c\) to \(c+\gamma\), shrinking the two-target region and expanding the one-target region.}
\label{fig:treaty-AV2}
\end{figure}

\Cref{fig:treaty-AV2} illustrates the mechanism. In the left panel, without the treaty, the gray regions show what would happen if \(\Vtwo\) stayed out: some types would attack only \(\Vone\), while others would attack both. Because the two-target region is large enough, \(\Vtwo\) joins and the actual attack set is the red CR wedge. In the right panel, the treaty raises the cutoff for attacking \(\Vtwo\) from \(c\) to \(c+\gamma\). The two-target region shrinks, the one-target region expands, and the first attack becomes easier to interpret as limited. A treaty that appears to protect the bystander can therefore make the first victim easier to attack.

The German--Soviet non-aggression pact (the Molotov--Ribbentrop Pact) is a useful historical analogy. The point is not that the pact literally guaranteed Soviet safety: Germany broke the pact and invaded the Soviet Union in 1941. Rather, the pact temporarily made German aggression against Poland easier to interpret as a limited conflict, thereby weakening the inference that the Soviet Union itself was immediately next \citep{USHMMGermanSovietPact}. This is the equilibrium danger of a neutrality pact in the model. It reassures the bystander exactly when reassurance exposes the first victim.

A collective-defense treaty works in the opposite direction.  It does not have to change the
posterior probability \(\phi(c)\).  Instead, it changes the payoff from refusing to join.  If
\(\Vone\) fights and \(\Vtwo\) stays out, \(\Vtwo\) now bears the expected future loss plus the
cost \(\gamma\) of violating the commitment.  This commitment can make joining optimal even
when the ordinary D\&C posterior is below the belief threshold.

\begin{proposition}[Collective-defense treaties]\label{prop:treaty-12}
Fix \(c>\hatc\). If \(\gamma>\ell[\kappa-\phi(c)]\), then \DC\ is not sustainable under a treaty between \(\Vone\) and \(\Vtwo\). Under \Cref{ass:positive-selection}, collective resistance is the unique pure regular equilibrium.
\end{proposition}

Without the treaty and with \(c>\hatc\), the posterior \(\phi(c)\) is below the belief threshold, so \(\Vtwo\) stays out. With the treaty, staying out after \(\Vone\) fights imposes the additional cost \(\gamma\). The inequality in the proposition says that this cost is large enough to make joining optimal even after the ordinary D\&C posterior.

The selected attack set then shrinks from \(\{b_1\ge c\}\) to \(J(c)\), so the probability of an attack on \(\Vone\) strictly falls. The probability that \(\Vtwo\) is attacked also falls whenever \(\mu(J(c))<\Pr(b_2\ge c)\), for example when \(p_2\) is sufficiently small. This is the formal version of Article 5-style collective defense. The institution removes the bystander's temptation to wait for more information. It does not have to persuade \(\Vtwo\) that \(\phi(c)\) is large; it makes refusal to join costly even when \(\phi(c)\) is small \citep{NATO1949}.

Treaties can also interact with the behavioral forces in \Cref{sec:behavioral}. The Molotov--Ribbentrop Pact gave Stalin a concrete reason to discount warnings of German attack; a CIA review of \emph{What Stalin Knew} describes how Stalin dismissed repeated warnings and the German military buildup as disinformation or provocation until the morning of June 22, 1941 \citep{CIAWhatStalinKnew}. The Doha Agreement is another cautionary example. The agreement was between the United States and the Taliban, and it reflected the desire to end the war in Afghanistan; SIGAR later reported that Taliban attacks increased during the negotiation period and that, by April 2021, U.S. intelligence assessed the Taliban to be confident it could achieve military victory \citep{StateDohaAgreement,SIGARCollapse2022}. These examples are not literal applications of the formal treaty model. They illustrate a related behavioral margin: treaty-like assurances can become material for wishful or magical thinking when the desire to avoid further conflict is strong.

\subsection{More than two victims}\label{sec:more-than-two}

With more than two victims, a bystander asks not only whether she will be attacked later, but also whether someone else would protect her then. Protection behind a second line can therefore weaken her incentive to defend the first victim.

For simplicity, we assume that at most two victims can coordinate resistance at a time.\footnote{The downstream-protection effect can survive when all three victims can resist together. \Cref{oa:three-victims} gives an example with strict incentives in which a three-victim coalition is feasible and more effective than a pair, yet the second victim waits for protection later rather than joining the resistance to the first attack.} Consider three victims with \(1\ge b_1\ge b_2\ge b_3\ge0\), attacked in that order. Initially, \(\Vtwo\) can join \(\Vone\)'s resistance, and \(V_3\) can join \(\Vtwo\)'s resistance later. The latter pairing is a fallback alliance, and \(t\in[0,1]\) measures its credibility: after an attack on \(\Vtwo\), the alliance turns out to be available with probability \(t\), independently of benefits. If it is available, \(V_3\) is committed to joining \(\Vtwo\)'s resistance; if not, \(V_3\) cannot join. Availability is revealed after the attacker has chosen to attack \(\Vtwo\). Thus the attacker anticipates protection but cannot condition that attack on its realization. We retain the baseline victim payoffs and assume \(p_2\le1/2\). The joint distribution of \((b_1,b_2,b_3)\) is affiliated, and its marginal on \((b_1,b_2)\) is the baseline distribution.\footnote{The attacker cannot bypass a victim. Coalition members leave the game after a joint fight, while remaining outsiders stay reachable. \Cref{oa:downstream-protection} specifies the continuation rules and the affiliation condition.}

If \(\Vtwo\) stays out initially, the attacker subsequently attacks her precisely for types in
\[
K_t(c):=\left\{(b_1,b_2,b_3):
(1-t)\big[b_2-c+(b_3-c)_+\big]
+t\big[p_2(b_2+b_3)-2c\big]\ge0\right\}.
\]
Here the first term is the attacker's payoff when the alliance is unavailable (\(\Vtwo\) surrenders, and \(V_3\) is attacked next if she is a profitable target), and the second is the payoff from a fight against the protected pair. Let \(\hat c\) be the baseline cost threshold in \Cref{prop:cost}. Without downstream protection, costs below \(\hat c\) yield collective resistance. Protection of the second victim can turn part of this region into \DC.

\begin{proposition}[Downstream protection expands the \DC\ region]\label{prop:three-victims}
Under attacker-preferred pure selection, stronger downstream protection expands the \DC\ region. For each \(c<\hat c\), there is a unique \(\hat t(c)\in(0,1)\) such that a collective-resistance equilibrium occurs for \(t<\hat t(c)\), with \(\Vone\) and \(\Vtwo\) resisting together if attacked, while a \DC\ equilibrium occurs for \(t>\hat t(c)\). For \(c>\hat c\), a \DC\ equilibrium occurs for every \(t\).
\end{proposition}

Stronger fallback protection makes waiting safer in two ways: it reduces \(\Vtwo\)'s loss if she is attacked and can deter that later attack altogether. The cutoff equates her loss from waiting to the loss from immediate pair defense:
\[
\big[1-(1-\kappa)\hat t(c)\big]
\Pr\!\big(K_{\hat t(c)}(c)\mid b_1\ge c\big)=\kappa.
\]
The first factor is \(\Vtwo\)'s expected loss, as a fraction of \(\ell\), if she waits and is then attacked (\(\kappa\) when the alliance is available and one otherwise); the second is the probability of that attack. Here division means that the first victim is isolated; the bystander may still receive help if attacked later. \Cref{fig:ct-partition} illustrates the resulting partition.

\begin{figure}[htb]
\centering
\includegraphics[width=.72\textwidth]{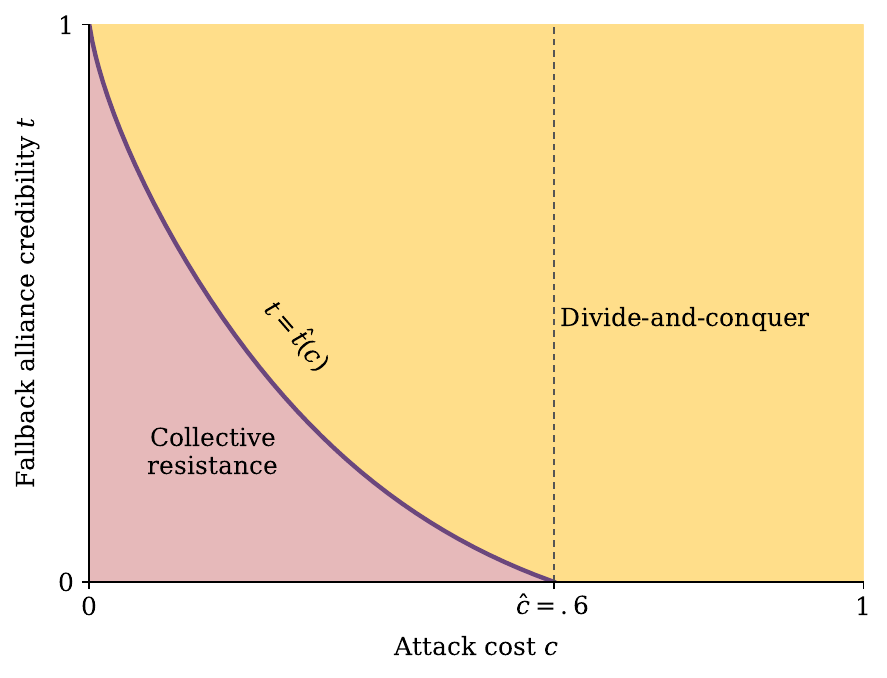}
\caption{Downstream protection can isolate the first victim. The red region has collective resistance between \(\Vone\) and \(\Vtwo\); the yellow region has surrender by \(\Vone\) under attacker-preferred pure selection. The illustration uses a uniform baseline distribution of \((b_1,b_2)\), with \(b_3\) conditionally uniform on \([0,b_2]\), \(p_1=.9\), \(p_2=.1\), and \(d/\ell=.15\). Thus the belief threshold is \(\kappa=.25\) and \(\hat c=.6\). The curve is computed from the attacker's continuation decision and the displayed indifference condition.}
\label{fig:ct-partition}
\end{figure}

A bystander may thus wait because a later attack would activate another resistance coalition. The Russia--Ukraine war illustrates the force captured by the model. Ukraine was outside NATO's Article~5, while many neighboring states were protected by it; Finland joined NATO in 2023 and Sweden in 2024 \citep{NATOFinland2023,NATOSweden2024}. From the perspective of the model, such downstream protection can create a \DC\ expectation: protected neighbors may perceive a later attack on themselves as less likely or less dangerous because it would trigger a broader collective-defense system, and for that reason an aggressor may expect them to be less willing to enter the first victim's fight \citep{NATO1949}.\footnote{We do not claim that NATO enlargement caused the invasion, nor that protected states were indifferent to Ukraine. The narrower point is that downstream protection can shape both bystander behavior and aggressor expectations.} That the invasion nevertheless met resistance from Ukraine, backed by broad Western support, is not at odds with this reading: as the mixed equilibrium of \Cref{sec:interp} shows, an attacker who expects acquiescence may be surprised by joint resistance.

\section{Conclusion}\label{sec:conclusion}

Divide-and-conquer succeeds when a first attack is sufficiently consistent with a target-specific motive. The victims may understand the value of joint resistance, yet fail to resist jointly because the bystander is unsure whether the attack reveals a threat to her. The model shows how this inferential friction can generate sequential conquest and uncertain collective resistance. Higher attack costs can make division sustainable and trigger an upward jump in attacks under both an attacker-preferred and a mixed benchmark. The conditional danger to the second victim falls with cost under pure division but remains at the belief threshold under mixing. Behavioral responses, rhetoric, treaties, and downstream protection affect the same collective-response problem. Effective collective defense requires not only the ability to fight together, but also a common interpretation of what the first attack reveals.

\appendix

\section{Proofs for the Baseline Model}\label{app:proofs}

\begin{proof}[Proof of \Cref{thm:baseline}]
In any pure regular equilibrium, \(\Vtwo\)'s action after the history in which \(\A\) attacks \(\Vone\) and \(\Vone\) fights is either to stay out or to join. If she stays out, \(\Vone\)'s fight would be solitary and would yield loss \(p_1\ell+d>\ell\); hence \(\Vone\) surrenders when attacked. Given surrender by \(\Vone\), \(\A\)'s payoff from the initial attack is \(b_1-c+\max\{b_2-c,0\}\), so, on \(B\), \(\A\) attacks \(\Vone\) if and only if \(b_1\ge c\), and later attacks \(\Vtwo\) if and only if \(b_2\ge c\). Regularity then gives posterior \(\Pr(\cdot\mid b_1\ge c)\) after the off-path fight by \(\Vone\). Staying out gives \(\Vtwo\) expected loss \(\phi(c)\ell\), while joining gives loss \(p_2\ell+d=\kappa\ell\). Thus the \DC\ equilibrium exists if and only if \(\phi(c)\le\kappa\).

If instead \(\Vtwo\) joins after \(\Vone\)'s fight, then \(\Vone\)'s loss from fighting is \(p_2\ell+d<\ell\), so \(\Vone\) fights when attacked. Anticipating joint resistance, \(\A\)'s payoff from attacking is \(p_2(b_1+b_2)-2c\), and hence \(\A\) attacks if and only if \((b_1,b_2)\in J(c)\). If \(\mu(J(c))=0\), no positive mass of types attacks and collective resistance fully deters attack; \(\Vtwo\)'s joining after an off-path attack is then supported by beliefs concentrated on types with \(b_2\ge c\), which is the convention \(\psi(c)=1\). If \(\mu(J(c))>0\), conditional on such an attack, \(\Vtwo\)'s posterior probability of being a future target is \(\psi(c)\). Joining is optimal if and only if \(\psi(c)\ell\ge p_2\ell+d\), equivalently \(\psi(c)\ge\kappa\). These two cases exhaust pure regular equilibria.

The trichotomy follows from \Cref{ass:positive-selection}. If \(\phi(c)>\kappa\), \DC\ fails and \(\psi(c)\ge\phi(c)>\kappa\), so collective resistance exists. If \(\psi(c)<\kappa\), collective resistance fails and \(\phi(c)\le\psi(c)<\kappa\), so \DC\ exists. If \(\phi(c)\le\kappa\le\psi(c)\), both exist.
\end{proof}

\begin{proof}[Proof of \Cref{lem:LRphi}]
On \(B\), \(b_2\ge c\) implies \(b_1\ge c\), so \(\phi(c)=\Pr(b_2\ge c)/\Pr(b_1\ge c)\). The relative-survival inequality in \Cref{ass:LR} is exactly \(\phi(c')\le\phi(c)\). At \(c=0\), both events have probability one, so \(\phi(0)=1\). Where marginal densities exist,
\[
\phi'(c)=\phi(c)\left[\frac{f_1(c)}{\Pr(b_1\ge c)}-\frac{f_2(c)}{\Pr(b_2\ge c)}\right],
\]
which also verifies the equivalent hazard ordering.
\end{proof}

\begin{proof}[Proof of \Cref{prop:cost}]
Continuity, strict decrease, and the endpoint conditions \(\phi(0)=1>\kappa>\lim_{c\uparrow1}\phi(c)\) give the unique threshold. Below it, \Cref{cor:unique-CR} applies. The sets \(J(c)\) shrink strictly in probability for \(0<c<p_2\) and have zero probability for \(c\ge p_2\). Attacker-preferred selection switches to \(D(c)=\{b_1\ge c\}\) above the threshold. For every interior cost, \(J(c)\) is a strict subset of \(D(c)\) in probability.

For the mixed benchmark, consider \(\beta=\beta_0\) and small positive \(\alpha\), as defined in \Cref{oa:mixed}. At \(\alpha=0\), the attack set is \(D(c)\). At \(c=\hat c\), the association argument in the proof of \Cref{cor:unique-CR} strictly raises the threat posterior above \(\kappa\) whenever \(\alpha>0\) is sufficiently small. To see strictness, recall from that proof that resistance retains a fraction of each \(b_2\)-row of the surrender pool \(D(\hat c)\), and that this fraction is nondecreasing in \(b_2\). For small \(\alpha\), a positive band of rows with \(b_2>\hat c\) remains fully admitted, while resistance removes a positive mass of types below \(b_2=\hat c\). The retention fractions are therefore not constant across the threat cutoff, and the posterior rises strictly. By continuity, the posterior remains above \(\kappa\) at any fixed such \(\alpha\) for nearby costs, whereas at \(\alpha=0\) it is below \(\kappa\) for costs above \(\hat c\). The first mixed root consequently approaches zero resistance as \(c\downarrow\hat c\), and its initial attack set approaches \(D(\hat c)\). This proves the two one-sided limits and the strict jump. The same limits give the probability that \(\Vtwo\) is attacked, discussed in the text: it is \(\mu(J(\hat c))\) on the left and, because resistance vanishes, \(\kappa\Pr(b_1\ge\hat c)=\Pr(b_2\ge\hat c)\) on the right under both benchmarks. Zero mass of \(J(\hat c)\) is equivalent to \(\hat c\ge p_2\). Above the threshold, attacker-preferred attack probability is \(1-F_1(c)\), which strictly decreases. The mixed decline in part (iii) is proved in \Cref{oa:mixed-statics}. The conditional second-target comparison follows from \Cref{prop:mixed} and the monotonicity of \(\phi\); under \Cref{def:ordered}, its limit is zero.
\end{proof}

\begin{proof}[Proof of \Cref{prop:correlation}]
The fraction in \Cref{eq:copulaphi} is increasing in \(C_\rho(F(c),F(c))\). Where derivatives exist,
\[
\frac{\partial\phi_\rho(c)}{\partial\rho}
=\frac{2[1-F(c)]}{[1-C_\rho(F(c),F(c))]^2}
\frac{\partial C_\rho(F(c),F(c))}{\partial\rho}\ge0.
\]
The same monotone comparison applies without differentiability. The initial \DC\ attack probability is \(1-C_\rho(F(c),F(c))\), which weakly decreases, giving part (i). The example in \Cref{oa:correlation-mixed} proves part (ii).
\end{proof}

\begin{proof}[Proof of \Cref{prop:asymmetry}]
On the support \(B\), \(b_2\ge c\) implies \(b_1\ge c\), and hence implies \(b_1+a\ge c\) for every \(a\ge0\). Thus
\[
  \phi_a(c)=\frac{\Pbb(b_2\ge c)}{\Pbb(b_1+a\ge c)}=
  \frac{\Pbb(b_2\ge c)}{\Pbb(b_1\ge c-a)}.
\]
The numerator is independent of \(a\), while the denominator is weakly increasing in \(a\). Hence \(\phi_a(c)\) is weakly decreasing in \(a\), and \DC\ attack probability \(\Pr(b_1\ge c-a)\) weakly increases. The mixed comparison in part (ii) is proved in \Cref{oa:specialness-mixed}.
\end{proof}

\section{Proofs for the Extensions}\label{app:extension-proofs}

\begin{proof}[Proof of \Cref{prop:wishful-thinking}]
If \(\Vtwo\) stays out, the most reassuring feasible subjective posterior is \(m=\zeta\phi(c)\). Hence the optimized loss from staying out is \((1+\theta\zeta)\phi(c)\ell\), while the loss from joining is \((1+\theta)\kappa\ell\). Thus staying out is optimal if and only if
\[
        \phi(c)\le \kappa^{WT}(\theta,\zeta):=\frac{(1+\theta)\kappa}{1+\theta\zeta}.
\]
Strict monotonicity of \(\phi\) gives the cutoff characterization in the interior case. If \(\theta>0\) and \(\zeta<1\), then \(\kappa^{WT}(\theta,\zeta)>\kappa\), so the cutoff is strictly below the Bayesian cutoff \(\hatc\).
\end{proof}

\begin{proof}[Proof of \Cref{prop:magical-thinking}]
Let \(\Phi_\omega(c):=\Pr(b_2\ge c+\omega\mid b_1\ge c)\). Under the magical belief, staying out exposes \(\Vtwo\) to later attack only when \(b_2\ge c+\omega\). The perceived loss from staying out is therefore \(\Phi_\omega(c)\ell\), while joining costs \(\kappa\ell\). Thus \DC\ is sustained whenever \(\Phi_\omega(c)\le\kappa\). Since \(\Phi_\omega(c)\le\phi(c)\) for every \(c\), every cost that sustains \DC\ in the Bayesian benchmark also sustains \DC\ under magical thinking.

At \(c=\hatc\), positive density on the interior of \(B\) and \(\omega>0\) imply \(\Phi_\omega(\hatc)<\phi(\hatc)=\kappa\). By continuity, there is \(\varepsilon>0\) such that \(\Phi_\omega(c)<\kappa\) for every \(c\in[\hatc-\varepsilon,\hatc]\). For \(c\ge\hatc\), \(\Phi_\omega(c)\le\phi(c)\le\kappa\). Define
\[
\hat c^{MT}_\omega:=\inf\{x\in[0,1]: \Phi_\omega(c)\le\kappa \text{ for every } c\in[x,1]\}.
\]
Then \(\hat c^{MT}_\omega<\hatc\), and \DC\ is sustained for every \(c\ge\hat c^{MT}_\omega\).
\end{proof}

\begin{proof}[Proof of \Cref{prop:propaganda-formal}]
For a signal cell \(E\) with positive probability, define
\[
        \phi_E(c):=\Pr(b_2\ge c\mid b_1\ge c,E),
        \qquad
        \psi_E(c):=\Pr(b_2\ge c\mid J(c),E).
\]
When a collective-resistance cell has zero attack probability, use \(\psi_E(c)=1\), implemented by beliefs on sufficiently high-threat types in that cell. A cell assigned to \DC\ is checked with \(\phi_E(c)\le\kappa\); a cell assigned to collective resistance is checked with \(\psi_E(c)\ge\kappa\).

Fix \(c<\hatc\). For \(\Delta<1-c\), the cell \(L_\Delta\cap\{b_1\ge c\}\) has positive probability. As \(\Delta\uparrow1-c\), the event \(L_\Delta\cap\{b_1\ge c,b_2\ge c\}\) shrinks to a boundary set, while \(L_\Delta\cap\{b_1\ge c\}\) retains positive probability. Hence \(\phi_{L_\Delta}(c)\to0\). Thus there exists \(\bar\Delta_1(c)<1-c\) such that \(\phi_{L_\Delta}(c)<\kappa\) for every \(\Delta\in(\bar\Delta_1(c),1-c)\). At \(\Delta=1-c\), the cell \(H_\Delta\) contains all types with \(b_2\ge c\); it removes, if anything, only types with \(b_2<c\). Therefore \(\psi_{H_{1-c}}(c)\ge\psi(c)\). By \Cref{ass:positive-selection} and \(c<\hatc\), \(\psi(c)\ge\phi(c)>\kappa\). By continuity, there exists \(\bar\Delta_2(c)<1-c\) such that \(\psi_{H_\Delta}(c)>\kappa\) for every \(\Delta\in(\bar\Delta_2(c),1-c)\). Likewise \(\phi_{H_{1-c}}(c)\ge\phi(c)>\kappa\), so, after raising \(\bar\Delta_2(c)\) if necessary, \DC\ is not sustainable on \(H_\Delta\) for these \(\Delta\), and collective resistance is the selected continuation there. Let \(\bar\Delta(c)=\max\{\bar\Delta_1(c),\bar\Delta_2(c)\}\). Then every \(\Delta\in(\bar\Delta(c),1-c)\) supports \DC\ on \(L_\Delta\) and collective resistance on \(H_\Delta\).

Now fix \(c>\hatc\). Since \(L_\Delta\) converges to the whole support up to the diagonal, which has measure zero, \(\phi_{L_\Delta}(c)\to\phi(c)<\kappa\) as \(\Delta\downarrow0\). Hence there exists \(\underline\Delta_1(c)>0\) such that \(\phi_{L_\Delta}(c)<\kappa\) for every \(\Delta\in(0,\underline\Delta_1(c))\). Choose \(\underline\Delta_2(c)>0\) such that \(p_2(2c+\Delta)<2c\) for every \(\Delta\in(0,\underline\Delta_2(c))\). If \((b_1,b_2)\in H_\Delta\), \(b_2<c\), and \(b_1\le b_2+\Delta\), then \(p_2(b_1+b_2)<p_2(2c+\Delta)<2c\); hence such a type is not in \(J(c)\). Therefore \(J(c)\cap H_\Delta\subseteq\{b_2\ge c\}\), so \(\psi_{H_\Delta}(c)=1\) whenever the conditioning event is nonempty. Moreover, \DC\ is not sustainable on \(H_\Delta\) when \(\Delta\) is small: within \(H_\Delta\cap\{b_1\ge c\}\), every type with \(b_2<c\) has \(b_1<c+\Delta\), and the share of such types vanishes as \(\Delta\downarrow0\) provided the density \(f\) remains continuous and strictly positive at the diagonal (as it does under \Cref{def:ordered}); hence \(\phi_{H_\Delta}(c)>\kappa\) after lowering \(\underline\Delta_2(c)\) if necessary. Let \(\underline\Delta(c)=\min\{\underline\Delta_1(c),\underline\Delta_2(c)\}\). Then every \(\Delta\in(0,\underline\Delta(c))\) sustains collective resistance on \(H_\Delta\) and leaves \DC\ sustained on \(L_\Delta\). The attack-set description in the text follows from the two cellwise continuations.
\end{proof}

\begin{proof}[Proof of \Cref{prop:treaty-A2}]
Under the treaty, if \(\Vone\) is expected to surrender, \(\A\)'s payoff from attacking \(\Vone\) is \(b_1-c+(b_2-c-\gamma)_+\), so \(\A\) attacks \(\Vone\) if and only if \(b_1\ge c\). Conditional on \(\Vone\)'s surrender, attacking \(\Vtwo\) yields \(b_2-c-\gamma\), so \(\A\) attacks \(\Vtwo\) if and only if \(b_2\ge c+\gamma\). Thus \(\Vtwo\)'s posterior probability of being a future target is \(\Pr(b_2\ge c+\gamma\mid b_1\ge c)\). If \(\Vtwo\) joins, she pays \(p_2\ell+d+\gamma\). The displayed condition makes staying out optimal. Given that \(\Vtwo\) stays out, \(\Vone\)'s resistance would be solitary, which is worse than surrender by \Cref{ass:victims}. Thus the treaty supports the \DC\ equilibrium, which is selected by attacker preference.

Without the treaty and with \(c<\hatc\), \DC\ is not sustainable, and the selected collective-resistance attack set is \(J(c)\). Since \(J(c)\subseteq\{b_1\ge c\}\) on \(B\), with strict containment under the maintained density conditions, the treaty-induced \DC\ equilibrium raises the probability of an attack on \(\Vone\). For the claim about attacks on \(\Vtwo\), choose \(\gamma<1-c\) so that the displayed condition holds; this is possible by taking \(\gamma\) close enough to \(1-c\). For sufficiently small \(p_2\), the set \(J(c)\) has arbitrarily small measure, and is empty if \(p_2<c\). Since \(\Pr(b_2\ge c+\gamma)>0\) whenever \(c+\gamma<1\), the probability that \(\Vtwo\) is attacked is also strictly larger for such parameters.
\end{proof}

\begin{proof}[Proof of \Cref{prop:treaty-12}]
Without the treaty and with \(c>\hatc\), \(\phi(c)<\kappa\), so \(\Vtwo\) stays out in the attacker-preferred \DC\ equilibrium. With the treaty, staying out after \(\Vone\) fights costs \(\phi(c)\ell+\gamma\), while joining costs \(p_2\ell+d=\kappa\ell\). If \(\gamma>\ell[\kappa-\phi(c)]\), joining is strictly optimal even after the attack-on-\(\Vone\) posterior. Thus \DC\ fails. Since \(\psi(c)\ge\phi(c)\) by \Cref{ass:positive-selection}, the same inequality implies \(\psi(c)\ell+\gamma>p_2\ell+d\), so collective resistance is sequentially rational. Anticipating collective resistance, \(\A\) attacks if and only if \((b_1,b_2)\in J(c)\).

The attack-probability comparisons follow from set inclusion. Under \DC, the initial attack set is \(\{b_1\ge c\}\), while under collective resistance it is \(J(c)\). On \(B\), \(J(c)\subseteq\{b_1\ge c\}\), with strict containment under the maintained density conditions, so attacks on \(\Vone\) strictly fall. Attacks on \(\Vtwo\) fall whenever \(\mu(J(c))<\Pr(b_2\ge c)\). A sufficient condition is \(p_2<2c/(1+c)\), because then \(J(c)\subseteq\{b_2\ge c\}\), with strict containment under positive density.
\end{proof}

\begin{proof}[Proof of \Cref{prop:three-victims}]
\Cref{oa:downstream-protection} derives the attacker's continuation decision, verifies the two pure initial responses, and proves the cutoff and attacker-preference claims.
\end{proof}

\bibliographystyle{plainnat}
\bibliography{DC-theory_refs}

\clearpage
\section*{Online Appendix}
\addcontentsline{toc}{section}{Online Appendix}
\setcounter{section}{0}
\setcounter{equation}{0}
\setcounter{lemma}{0}
\setcounter{example}{0}
\renewcommand{\thesection}{OA.\arabic{section}}
\renewcommand{\theequation}{OA.\arabic{equation}}
\renewcommand{\thelemma}{OA.\arabic{lemma}}
\renewcommand{\theexample}{OA.\arabic{example}}

\section{The maintained attack order}\label{app:v2first}
The baseline fixes \(\Vone\) as the first target. Two payoff comparisons explain why that order is useful for studying division. A simultaneous attack identifies both targets and yields the joint-fight payoff \(p_2(b_1+b_2)-2c\). A sequential \(\Vone\)-first attack yields that same payoff under collective resistance and more under surrender or lone fighting.

If an attack on \(\Vtwo\) first occurs with positive probability, rational attackers must have \(b_1\ge c\): a type with \(b_1<c\) also has \(b_2<c\) and cannot profit under any response. Thus \(\Vone\) knows she would be a profitable subsequent target and joins \(\Vtwo\)'s fight under regularity. Such an on-path attack produces joint resistance. It cannot exploit the bystander's hope of being spared that supports \DC\ when the higher-benefit victim is attacked first.

\section{Order-statistic microfoundation}\label{oa:orderstat}
For ordered independent draws from \(F\), the marginal densities are \(f_1(x)=2F(x)F'(x)\) and \(f_2(x)=2[1-F(x)]F'(x)\). Hence \(f_1(x)/f_2(x)=F(x)/[1-F(x)]\) is increasing, verifying likelihood-ratio order as well as \Cref{ass:LR}.

If \(b_1=\max\{X_1,X_2\}\) and \(b_2=\min\{X_1,X_2\}\) for i.i.d. draws from a distribution \(F\) with positive density on \((0,1)\), then
\[
\Pbb(b_2\ge c)=\Pbb(X_1\ge c,X_2\ge c)=[1-F(c)]^2,
\]
and
\[
\Pbb(b_1\ge c)=1-F(c)^2=[1-F(c)][1+F(c)].
\]
Therefore
\[
\phi(c)=\frac{1-F(c)}{1+F(c)},
\]
and differentiating gives
\[
\phi'(c)=\frac{-2F'(c)}{[1+F(c)]^2}<0.
\]
This benchmark also verifies \Cref{ass:positive-selection}.  Fix \(c\), and write
\(E_2=\{b_2\ge c\}\), \(E_1=\{b_1\ge c\}\), and
\(J=\{p_2(b_1+b_2)\ge 2c\}\).  Since \(p_2<1\) and \(b_2\le b_1\), the event
\(J\) is contained in \(E_1\).  Conditional on \(E_2\), the two primitive draws are i.i.d. from
\(F(\cdot\mid X\ge c)\).  Conditional on \(E_1\setminus E_2\), one draw is from
\(F(\cdot\mid X\ge c)\) and the other is from \(F(\cdot\mid X<c)\).  Hence the sum
\(X_1+X_2\) conditional on \(E_2\) first-order stochastically dominates the sum conditional
on \(E_1\setminus E_2\).  Since \(J\) is a high-sum event, \(\Pbb(J\mid E_2)\ge \Pbb(J\mid
E_1\setminus E_2)\).  Bayes' rule then gives
\[
  \Pbb(E_2\mid J)\ge \Pbb(E_2\mid E_1),
\]
which is \(\psi(c)\ge\phi(c)\).

\section[When A must defeat V1 to attack V2]{When \(\A\) must defeat \(\Vone\) to attack \(\Vtwo\)}
\label{oa:reachability}

The main text assumes that, if \(\Vtwo\) stays out after an off-path lone fight by \(\Vone\), \(\A\) can later attack \(\Vtwo\) regardless of the outcome of that fight. This appendix considers the alternative specification in which \(\A\) can attack \(\Vtwo\) only if she first defeats \(\Vone\) in the lone fight. Thus, if \(\A\) loses to \(\Vone\), the sequence ends and \(\Vtwo\) is not attacked.

Consider first a \DC\ candidate. After observing an attack on \(\Vone\), regularity gives the same posterior
\[
  \phi(c)=\Pr(b_2\ge c\mid b_1\ge c).
\]
If \(\Vtwo\) stays out after an off-path fight by \(\Vone\), she is attacked later only if two events occur: \(\A\) defeats \(\Vone\) in the lone fight, and \(\A\) finds \(\Vtwo\) profitable. Since \(\A\) wins a lone fight with probability \(p_1\), \(\Vtwo\)'s expected loss from staying out is
\[
  p_1\phi(c)\ell .
\]
Joining costs \(p_2\ell+d=\kappa\ell\). Hence, the \DC\ condition becomes
\[
  p_1\phi(c)\le \kappa .
\]
This condition is weaker than the main-text condition \(\phi(c)\le\kappa\). Outcome-dependent access to \(\Vtwo\) therefore makes \DC\ easier to sustain: by staying out, \(\Vtwo\) preserves the possibility that \(\Vone\)'s lone fight defeats \(\A\) and prevents any later attack on \(\Vtwo\).

Now consider a collective-resistance candidate. If \(\Vtwo\) joins after \(\Vone\) fights, then \(\Vone\) fights when attacked, and \(\A\) attacks only if
\[
  (b_1,b_2)\in J(c)=\{(b_1,b_2)\in B:p_2(b_1+b_2)\ge 2c\}.
\]
If \(\Vtwo\) deviates and stays out after such an attack, she is attacked later only if \(\A\) defeats \(\Vone\) in the lone fight and \(b_2\ge c\). Her expected loss from staying out is therefore
\[
  p_1\psi(c)\ell ,
\]
where \(\psi(c)=\Pr(b_2\ge c\mid (b_1,b_2)\in J(c))\) when \(\mu(J(c))>0\), with the same convention as in the main text when \(J(c)\) has zero probability. Thus collective resistance is sequentially rational if and only if
\[
  p_1\psi(c)\ge \kappa .
\]

The cost comparative static is unchanged in form. Under the same monotonicity conditions as in \Cref{sec:cost}, the D\&C cutoff is determined by
\[
  \phi(c)=\frac{\kappa}{p_1}
\]
instead of \(\phi(c)=\kappa\), whenever \(\kappa/p_1<1\). Since \(p_1<1\), this cutoff lies strictly below the main-text cutoff \(\hat c\). If \(\kappa\ge p_1\), then \(p_1\phi(c)\le\kappa\) for all \(c\), so \DC\ is sustained for all relevant costs.

The case \(\kappa<p_1\) preserves the full-information benchmark: a bystander who knows she is next still prefers to join immediately rather than rely on the possibility that \(\A\) loses a lone fight against \(\Vone\). If instead \(\kappa\ge p_1\), the alternative specification changes that benchmark, because even a bystander who knows she is next may prefer to stay out and let \(\Vone\)'s lone fight determine whether \(\A\) can proceed. In either case, making the second attack contingent on \(\A\)'s defeating \(\Vone\) does not weaken the \DC\ logic; it makes the bystander's inaction easier to support.
\section{Mixed equilibrium and uniqueness}\label{oa:mixed}

Let \(\alpha\) be the probability that \(V_1\) fights and \(\beta\) the probability that \(V_2\) joins a fight. The first victim is indifferent when the probability of help equals
\[
\beta_0=\frac{p_1+d/\ell-1}{p_1-p_2}\in(0,1).
\]
She surrenders for \(\beta<\beta_0\) and fights for \(\beta>\beta_0\). Given these responses, the initial attack payoff is
\begin{align}
&(1-\alpha)\{b_1-c+(b_2-c)_+\}\notag\\
&\quad+\alpha(1-\beta)\{p_1b_1-c+(b_2-c)_+\}
+\alpha\beta\{p_2(b_1+b_2)-2c\}.\label{eq:mixed-payoff}
\end{align}
Define \(B_{\alpha,\beta}(c)\) as the types for which this payoff is nonnegative. Cutoff types have probability zero. A positive-attack equilibrium with uncertain collective resistance has
\[
0<\alpha\le1,\qquad \beta_0\le\beta<1,\qquad
\alpha<1\Longrightarrow\beta=\beta_0,
\]
and \(\Pr(b_2\ge c\mid B_{\alpha,\beta}(c))=\kappa\). The posterior does not change upon observing a resistance choice by an uninformed victim.

\paragraph{The mixed benchmark.} Increase \(\alpha\) from zero to one while holding \(\beta=\beta_0\); then increase \(\beta\) from \(\beta_0\) to one while holding \(\alpha=1\). Both segments respect the first victim's best response. The selected mixed benchmark is the first pair on this path at which the bystander's posterior equals \(\kappa\), provided genuine mixing exists. This fixes the selection among potentially several roots. In the strict-coexistence region such a pair always exists. Where no genuine mixed root exists, the benchmark uses the attacker-preferred pure outcome. At the cost threshold we use the collective-resistance outcome; the jump statement compares one-sided limits and does not depend on this tie convention.

\paragraph{Smooth comparative statics.} We call a segment of the mixed benchmark smooth when its mixing
probabilities and attack masses are differentiable in the parameter, the posterior is differentiable, and its
indifference equation has a nonzero derivative with respect to the
mixing probability that varies along the victims' response path.
On the segment on which both victims randomize, $\beta=\beta_0$, and
this derivative is the derivative with respect to $\alpha$.
The selected root is the first crossing starting from a posterior below
$\kappa$.  Consequently that nonzero derivative is positive.

\begin{proof}[Proof of \Cref{prop:mixed}]
On the response path, the attack payoff decreases pointwise from the surrender payoff to the joint-fight payoff. Indeed, surrender is better for the attacker than lone fighting, and lone fighting is better than joint fighting, for every relevant interior type. Hence attack sets shrink continuously from \(D(c)=\{b_1\ge c\}\) toward \(J(c)\). The posterior is continuous whenever the attack pool has positive probability. Its initial value is \(\phi(c)<\kappa\). If \(\mu(J(c))>0\), its terminal value is \(\psi(c)>\kappa\). If \(\mu(J(c))=0\), before attacks disappear the pool concentrates near \((1,1)\), and its threat posterior approaches one. The pool has positive mass until its corner payoff reaches zero. Since \(c<1\) and \(\kappa<1\), a crossing occurs before that collapse. At that crossing \(\alpha>0\) and \(\beta<1\), so joint resistance has probability strictly between zero and one.

The payoff ordering gives \(J(c)\subseteq B_{\alpha,\beta}(c)\subseteq D(c)\). Interior positivity and genuine resistance give the strict probability inequalities: types near the low-threat surrender boundary are deterred, while types near a positive joint-war boundary enter; when joint resistance fully deters, the constructed attack pool still has positive mass.

Victim responses do not depend on the private type. Conditioning further on no joint resistance therefore leaves the posterior unchanged. The attacker subsequently attacks \(V_2\) exactly when \(b_2\ge c\), giving the final assertion.
\end{proof}

\paragraph{Payoff and attack rankings.} The attacker obtains the highest payoff under pure \DC, an intermediate payoff under genuine mixing, and the lowest under collective resistance, whenever those equilibria coexist. The ranking concerns optimized ex ante payoffs as well as each type's payoff. Victim 1's conditional expected loss in a mixed equilibrium is between \(\kappa\ell\) and \(\ell\). Multiplying by the strictly intermediate initial attack probability gives an ex ante loss between her collective-resistance and \DC\ losses. Victim 2's conditional loss is exactly \(\kappa\ell\) by indifference; her ex ante loss equals \(\ell\mu(B_{\alpha,\beta}(c)\cap\{b_2\ge c\})\), also strictly between the two pure losses. Thus the mixed benchmark has an intermediate payoff interpretation.

\paragraph{Actual attacks versus future-target types.} Conditional on an initial attack, collective resistance occurs with probability \(\alpha\beta\). Thus
\begin{align*}
\Pr(\text{both victims attacked}\mid\text{initial attack})
 &=\alpha\beta+(1-\alpha\beta)\kappa,\\
\Pr(\text{only }V_1\text{ attacked}\mid\text{initial attack})
 &=(1-\alpha\beta)(1-\kappa),\\
\Pr(\text{sequential attack on }V_2\mid\text{initial attack})
 &=(1-\alpha\beta)\kappa.
\end{align*}
Multiplying by \(\mu(B_{\alpha,\beta}(c))\) gives unconditional probabilities. The first two sum to total initial attacks.

\begin{proof}[Proof of \Cref{cor:unique-CR}]
For each response pair, the attack payoff is increasing in both benefits. Its horizontal zero-payoff boundary \(b_1=h(b_2)\) is nonincreasing in \(b_2\), and its attack set is contained in \(D(c)\). Conditional on \(D(c)\), the fraction of each \(b_2\)-row retained under resistance is
\[
\Pr(b_1\ge h(b_2)\mid b_2,\ b_1\ge c).
\]
This fraction is nondecreasing in \(b_2\): the threshold falls, and \Cref{ass:association} orders the conditional distributions upward. Clipping at the support boundaries preserves that order. Weighting the surrender pool by this increasing retention fraction weakly raises the posterior probability of \(b_2\ge c\), by the nonnegative covariance of increasing functions. Thus every nonempty resistance attack pool has threat posterior at least \(\phi(c)\). This also proves \(\psi(c)\ge\phi(c)\) when \(J(c)\) has positive mass; the zero-mass convention gives the remaining case.

If \(\phi(c)>\kappa\), an indifferent bystander is impossible in any positive-attack pool. She must join surely; victim 1 then fights surely and the attacker uses \(J(c)\). If no initial attack occurs, the outcome is already the fully deterred outcome. Since the maximal joint-war payoff is \(2p_2-2c\), full deterrence occurs exactly when \(c\ge p_2\), up to a probability-zero corner type.
\end{proof}

\paragraph{Primitive sufficient conditions.} A sufficient cost-independent density condition for \Cref{ass:association} is that
\[
\frac{f(x,y)}{\int_x^1 f(t,y)\,dt}
\quad\text{is nonincreasing in }y\quad\text{for each }x>y.
\]
For a row conditioned on \(D(c)\), survival from its lower endpoint \(\max\{c,y\}\) to a higher cutoff is the exponential of minus the integral of this conditional hazard. Raising \(y\) raises the lower endpoint and weakly lowers the hazard, so conditional survival increases. Affiliation, or the TP2 density inequality, implies this condition: for \(t\ge x>y_2>y_1\), it gives
\[
\frac{f(t,y_2)}{f(x,y_2)}\ge\frac{f(t,y_1)}{f(x,y_1)},
\]
and integrating over \(t\ge x\) orders the reciprocal hazards. The condition also holds for positive product-form densities on the ordered triangle. For ordered independent draws it follows directly from
\[
\Pr(b_1\ge x\mid b_2=y,\ b_1\ge c)
=\frac{1-F(\max\{x,c,y\})}{1-F(\max\{c,y\})}.
\]

\section{Cost comparative statics under mixing}\label{oa:mixed-statics}

When both victims mix, \(\beta=\beta_0\). The horizontal attack boundary is
\[
h(c,b_2)=
\begin{cases}
\dfrac{c(1+\alpha\beta_0)-\alpha\beta_0p_2b_2}{1-\alpha d/\ell},&b_2<c,\\[6pt]
\dfrac{2c-[1-\alpha\beta_0(1-p_2)]b_2}{1-\alpha d/\ell},&b_2\ge c.
\end{cases}
\]
The first boundary applies to types who would not attack \(V_2\) sequentially; the second includes that continuation. They meet at \(b_2=c\).

\begin{proof}[Proof of the general-prior part of \Cref{prop:cost}(iii)]
Let \(H\) and \(L\) be the masses of initial attackers with \(b_2\ge c\) and \(b_2<c\), respectively. These two masses are needed separately because mixed indifference is \((1-\kappa)H-\kappa L=0\). All derivatives below are partial derivatives holding the other argument, cost or resistance, fixed.

On active boundaries, \(b_2<h(c,b_2)<1\), both \(h_c\) and \(h_\alpha\) are positive. Their ratio is
\[
\frac{h_c}{h_\alpha}=
\begin{cases}
\dfrac{(1-\alpha d/\ell)(1+\alpha\beta_0)}{c(d/\ell+\beta_0)-\beta_0p_2b_2},&b_2<c,\\[6pt]
\dfrac{2(1-\alpha d/\ell)}{2cd/\ell-[d/\ell-\beta_0(1-p_2)]b_2},&b_2\ge c.
\end{cases}
\]
Both ratios increase with \(b_2\), with an upward jump at \(c\). The high-segment ordering uses
\[
d/\ell-\beta_0(1-p_2)
=\frac{(1-p_1)(1-p_2-d/\ell)}{p_1-p_2}>0.
\]
Keeping a marginal attacker indifferent after a cost increase requires a larger reduction in resistance for a higher second-target benefit. This ordering is independent of the density.

Differentiating the two masses along their moving boundaries gives
\begin{align*}
H_c&=-\int_{h(c,c)}^1 f(x,c)\,dx
-\int_c^1 f(h(c,y),y)h_c(c,y)\mathbf1_{\{y<h(c,y)<1\}}\,dy,\\
L_c&=+\int_{h(c,c)}^1 f(x,c)\,dx
-\int_0^c f(h(c,y),y)h_c(c,y)\mathbf1_{\{y<h(c,y)<1\}}\,dy,\\
H_\alpha&=-\int_c^1 f(h(c,y),y)h_\alpha(c,y)\mathbf1_{\{y<h(c,y)<1\}}\,dy,\\
L_\alpha&=-\int_0^c f(h(c,y),y)h_\alpha(c,y)\mathbf1_{\{y<h(c,y)<1\}}\,dy.
\end{align*}
An integral over an inactive row is zero. The ratio ordering implies a nonpositive boundary contribution to \(L_cH_\alpha-H_cL_\alpha\). Moving the classification threshold contributes
\[
\left(\int_{h(c,c)}^1 f(x,c)\,dx\right)(H_\alpha+L_\alpha)<0.
\]
At a genuine both-mixing root, \(L>0\) entails \(h(c,c)<1\), so the inequality is strict. Differentiating the mixed-indifference equation therefore gives
\[
\frac{d(H+L)}{dc}
=\frac{L_cH_\alpha-H_cL_\alpha}{(1-\kappa)H_\alpha-\kappa L_\alpha}.
\]
The denominator has the sign of the resistance derivative of the threat posterior. Total attacks thus fall when that posterior rises with resistance. A first non-tangent crossing starting below \(\kappa\) necessarily has that positive derivative.

For ordered independent draws, the proportional removal from an active row as resistance increases is
\[
h_\alpha(c,b_2)\frac{F'(h(c,b_2))}{1-F(h(c,b_2))}.
\]
The boundary and its resistance derivative decrease with \(b_2\). If the common hazard increases, this proportional removal therefore decreases with \(b_2\); on full rows it is zero. Stronger resistance selects toward higher second-target benefits and strictly raises the posterior at every genuine both-mixing root. This primitive condition gives at most one both-mixing root at each cost. The selected smooth-segment decline above requires only the positive derivative at the first crossing and therefore applies to general priors.
\end{proof}

\paragraph{Scope of the cost result.} If only \(V_2\) mixes, attacks can rise with cost even under uniform benefits: with \(p_1=.35\), \(p_2=.25\), and \(d/\ell=.70\), increasing \(c\) from \(.12\) to \(.14\) changes \(\beta\) from \(.978259\) to \(.822902\) at \(\alpha=1\), and raises attack probability from \(.560697\) to \(.562790\). Both points lie in strict pure coexistence and are the selected first roots on the bystander-only segment.

\paragraph{Reproducing the figures.} For uniform benefits, integrate twice the horizontal attack width \(1-\max\{b_2,h(c,b_2)\}\), clipped at zero, separately below and above \(c\); impose that the upper mass is \(\kappa\) of their sum. For \(p_1=.8\), \(p_2=.2\), \(d/\ell=.4\), the selected root has both victims mixing throughout \(.25<c<1\). At \(c=.3\), \(\alpha=(175-\sqrt{10129})/122=.609484\), total attacks are \(.796118\), and collective resistance conditional on an attack is \(.203161\). At \(c=.4\), \(\alpha=9/11\), total attacks are \(96/185\), and collective resistance is \(3/11\). The replication package includes the exact piecewise integration and plotting code.

\section{Dependence and mixed selection}\label{oa:correlation-mixed}

\paragraph{Opposite attack responses.} Let underlying uniform benefits have copula \(C_\rho(u,v)=uv+\rho uv(1-u)(1-v)\), with \(0\le\rho<1\). The ordered density is \(f_\rho(x,y)=2[1+\rho(1-2x)(1-2y)]\), which is positive and affiliated: \(\partial_{xy}\log f_\rho=4\rho/[1+\rho(1-2x)(1-2y)]^2\ge0\). At the parameters of \Cref{fig:mixed-cost} and \(c=.3\), increasing \(\rho\) from zero to \(.1\) lowers collective resistance from \(.203161\) to \(.178760\) and raises mixed initial attacks from \(.796118\) to \(.809015\). The corresponding pure \DC\ attack probability falls from \(.91\) to \(.905590\).

These comparisons follow by integrating the density over the attack pool.
For a horizontal attack cutoff $h$, its row mass is
\[
2(1-h)\{1-\rho(1-2b_2)h\}.
\]
At these parameters the low-threat cutoff is below one for every $\alpha\in[0,1]$.
Above $c$, the cutoff intersects the diagonal at
$z=2c/[2-(2/3)\alpha]$; rows above $z$ are fully retained.
Integrate the row masses over $[0,c]$, $[c,z]$, and $[z,1]$, respectively,
using the cutoffs in \Cref{oa:mixed-statics}.
If $H,L$ denote the high- and low-threat attack masses, indifference is
$Q=(2/5)H-(3/5)L=0$ and total attacks are $M=H+L$.
Implicit differentiation at $\rho=0$ gives
\[
\frac{d\alpha}{d\rho}=-\frac{Q_\rho}{Q_\alpha}=-.718459\ldots,
\qquad
\frac{dM}{d\rho}=M_\rho-M_\alpha\frac{Q_\rho}{Q_\alpha}=.134201\ldots>0.
\]
The signs therefore hold throughout a neighborhood of independence.
At both reported values of $\rho$, there is exactly one indifference root
on the both-mixing segment, so each is the selected first root.
Indeed, multiplying $Q$ by its positive denominator gives a degree-six
polynomial; substituting $\alpha=t/(1+t)$ produces the coefficient signs
$(-,-,-,+,+,+,+)$ at both values. Descartes' rule and the opposite signs
of $Q(0)$ and $Q(1)$ establish uniqueness on $(0,1)$.

\section{Specialness under the mixed benchmark}\label{oa:specialness-mixed}
Fix the cost and let $a$ denote the commonly known additive first-target
benefit.  The surrender attack set is $D_a(c)=\{b_1\ge c-a\}$, and
\[
\phi_a(c)=\frac{\Pr(b_2\ge c)}{\Pr(b_1\ge c-a)},
\]
where the denominator is one if $a\ge c$.  Hence the posterior weakly
falls and the surrender attack set weakly expands with $a$.

\begin{proof}[Proof of part (ii) of \Cref{prop:asymmetry}]
Adding $a$ replaces $b_1$ by $b_1+a$ in every continuation payoff.
The horizontal attack boundary therefore becomes
\[
b_1=h(c,b_2;\alpha)-a,
\]
where $h$ is the boundary in the cost proof.  Write $H$ and $L$ for
the masses of initial attackers with $b_2\ge c$ and $b_2<c$,
respectively.  Mixed indifference is $(1-\kappa)H-\kappa L=0$.
Holding resistance fixed, raising $a$ admits marginal attackers at
unit speed on each active horizontal boundary.  Holding $a$ fixed,
raising resistance removes them at speed
\[
h_\alpha(c,y)=
\begin{cases}
\displaystyle\frac{c(\beta_0+d/\ell)-\beta_0p_2y}
 {(1-\alpha d/\ell)^2},&y<c,\\[6pt]
\displaystyle\frac{2cd/\ell-[d/\ell-\beta_0(1-p_2)]y}
 {(1-\alpha d/\ell)^2},&y\ge c.
\end{cases}
\]
This speed is positive on active boundaries and strictly decreases
with $y$, continuously across $c$.  The second slope uses
\[
d/\ell-\beta_0(1-p_2)
=\frac{(1-p_1)(1-p_2-d/\ell)}{p_1-p_2}>0.
\]
Let $w(y)=f(h(c,y;\alpha)-a,y)$ on active boundaries, with $w(y)=0$
otherwise.  Differentiating the masses and the indifference equation
gives
\[
\frac{d(H+L)}{da}
=\frac{L_aH_\alpha-H_aL_\alpha}
 {(1-\kappa)H_\alpha-\kappa L_\alpha}.
\]
The numerator equals
\[
\int_0^c\!\int_c^1
w(y)w(z)\big[h_\alpha(c,y)-h_\alpha(c,z)\big]\,dz\,dy\ge0.
\]
The denominator is positive at the selected smooth first crossing.
Initial attack probability therefore weakly increases with $a$,
strictly when both groups have active attack boundaries.
When the high-threat rows are all fully admitted, their mass is fixed,
and mixed indifference can instead keep total attacks constant.
The proof imposes no distributional ordering.
\end{proof}

\paragraph{When only the bystander randomizes.}
The both-victims-mixing restriction matters.  With uniform benefits,
$p_1=.35$, $p_2=.25$, $d/\ell=.70$, and $c=.12$, the selected root has
$\alpha=1$.  Raising $a$ from zero to $.001$ changes $\beta$ from
$.978258960$ to $.979679541$ and lowers initial attack probability
from $.560696681$ to $.560266394$.  Both parameter values satisfy
strict coexistence of the pure equilibria.  The conditional sequential
threat remains $\kappa=.95$ throughout.

These figures can be checked without numerical integration.
For this example let $A=.35-.10\beta$ and $K=1-.75\beta$.
In the relevant geometry,
\[
H=(1-c)^2-\frac{[2c-Aa-(A+K)c]^2}{A(A+K)},\qquad
L=2c(1+a)-\frac{c^2[2+\beta(2-p_2)]}{A}.
\]
The indifference equation $(1-\kappa)H-\kappa L=0$ is equivalent,
at $a=0$ and $a=.001$, respectively, to
\begin{align*}
36958\beta^2-90131\beta+52803&=0,\\
295827204\beta^2-721859228\beta+423264049&=0.
\end{align*}
Each equation has a unique root in $[\beta_0,1]$, with
$\beta_0=.5$.  At $a=0$, the posterior at $(\alpha,\beta)=(1,.5)$
is $.875066208$ and at $(1,1)$ is $.955489614$;
at $a=.001$ these values are $.874851739$ and $.955124366$.
Thus each root is indeed selected on the second part of the response
path, rather than the both-victims-mixing part.

\section[Downstream protection with a committed fallback alliance]{Downstream protection with a committed\\ fallback alliance}\label{oa:downstream-protection}

The extension retains sequential attacks and pair defense. Benefits have a continuous density, strictly positive on the interior of
\[
\{(b_1,b_2,b_3):0\le b_3\le b_2\le b_1\le1\},
\]
and their \((b_1,b_2)\) marginal satisfies the maintained baseline cost assumptions. The joint prior is affiliated: its density is multivariate totally positive of order two. This condition implies association after conditioning on \(b_1\ge c\), which is the property used below \citep[Appendix, Theorems 23--24]{MilgromWeber1982}. The victim-payoff assumption is unchanged, and \(p_2\le1/2\).

\paragraph{Continuation rules.} The attacker begins with \(\Vone\). If \(\Vone\) fights, only \(\Vtwo\) can join her. After surrender or a lone fight, the attacker can attack \(\Vtwo\) or stop; stopping ends the game. After an initial pair fight, only \(V_3\) remains. Outsiders are reachable after a fight regardless of its outcome, and coalition members cannot be attacked again or defend outsiders later. The attacker cannot bypass \(\Vtwo\) to attack \(V_3\). If \(\Vtwo\) is attacked, an independent draw makes the committed fallback alliance available with probability \(t\). The draw is observed after the attack and before \(\Vtwo\)'s resistance decision. When available, \(V_3\) must join a fight initiated by \(\Vtwo\); otherwise she cannot join. A supported \(\Vtwo\) fights because her loss is \(\kappa\ell<\ell\), while an unsupported \(\Vtwo\) surrenders because lone fighting is worse. A victim attacked alone at the last stage also surrenders.

\begin{proof}[Proof of \Cref{prop:three-victims}]
Without fallback support, the attacker's payoff from attacking \(\Vtwo\) is
\(b_2-c+(b_3-c)_+\). With support, it is \(p_2(b_2+b_3)-2c\). Because availability is not known when this attack is chosen, their probability-weighted average gives exactly the attack set \(K_t(c)\) in the text.

If \(b_2<c\), both payoffs are negative, so \(K_t(c)\subseteq\{b_2\ge c\}\). For \(b_2\ge c\), protection weakly lowers the attack payoff. If \(b_3\ge c\), the payoff difference is \((1-p_2)(b_2+b_3)>0\); if \(b_3<c\), it is
\[
(1-p_2)b_2-p_2b_3+c\ge2(1-p_2)c>0.
\]
Consequently, \(K_t(c)\) shrinks as \(t\) rises, and \(K_0(c)=\{b_2\ge c\}\).

If \(\Vtwo\) stays out initially, \(\Vone\) surrenders. The attacker then obtains
\[
b_1-c+
\left((1-t)\big[b_2-c+(b_3-c)_+\big]
+t\big[p_2(b_2+b_3)-2c\big]\right)_+.
\]
This is nonnegative exactly when \(b_1\ge c\). In particular, if \(b_1<c\), ordering makes every continuation payoff negative. Thus the initial-attack posterior remains the prior conditional on \(b_1\ge c\). Regularity retains that posterior after an unexpected lone fight by \(\Vone\).

Conditional on a later attack, the normalized loss to \(\Vtwo\) is
\((1-t)+t\kappa=1-(1-\kappa)t\). Independence and the timing of availability therefore give her normalized loss from waiting as
\[
\big[1-(1-\kappa)t\big]\Pr\!\big(K_t(c)\mid b_1\ge c\big).
\]
The expression is continuous in \(t\), since the boundary of \(K_t(c)\) has zero prior probability, and it strictly decreases wherever it is positive. At \(t=0\), it equals \(\phi(c)\). At \(t=1\), it is at most \(\kappa\phi(c)<\kappa\), where strict inequality follows from positive interior density and \(0<c<1\). Hence divide-and-conquer exists exactly when the displayed loss is at most \(\kappa\). For \(c>\hat c\), this holds for every \(t\). For \(c<\hat c\), it crosses \(\kappa\) exactly once, at \(\hat t(c)\in(0,1)\).

It remains to verify collective resistance whenever division fails. If \(\Vone\) and \(\Vtwo\) fight together, the attacker's optimized initial payoff is
\[
p_2(b_1+b_2)-2c+(b_3-c)_+.
\]
Its nonnegative-payoff event is an increasing subset of \(\{b_1\ge c\}\). The later-attack event \(K_t(c)\) is also increasing in benefits. Affiliation is preserved by conditioning on \(b_1\ge c\), and association consequently implies
\[
\Pr\!\left(K_t(c)\,\middle|\,
p_2(b_1+b_2)-2c+(b_3-c)_+\ge0\right)
\ge
\Pr\!\big(K_t(c)\mid b_1\ge c\big)
\]
whenever the first conditioning event has positive probability. Thus if division fails, waiting costs more than \(\kappa\) even after an attack drawn from the collective-resistance pool. Joining is optimal, and \(\Vone\) fights in anticipation of that support. If the collective-resistance attack event has zero probability, beliefs conditional on \(b_1\ge c\) after an unexpected attack sustain the same participation decision and full deterrence. These are the two possible pure initial outcomes: without \(\Vtwo\)'s support, \(\Vone\) surrenders; with support, she fights.

Finally, the attacker weakly prefers division whenever it is sustainable. Her optimized continuation payoff under division is at least
\(\big[p_2(b_2+b_3)-2c\big]_+\). If \(b_3\ge c\), her initial division payoff is therefore at least
\[
b_1-c+p_2(b_2+b_3)-2c
=p_2(b_1+b_2)-2c+(b_3-c)
 +(1-p_2)(b_1-b_3),
\]
which is no smaller than her pair-resistance payoff. If \(b_3<c\), her division payoff is at least \(b_1-c\), while
\[
p_2(b_1+b_2)-2c-(b_1-c)
\le(2p_2-1)b_1-c<0
\]
under \(p_2\le1/2\). Taking the option not to attack preserves the weak ranking for every type. Attacker-preferred pure selection therefore gives the stated partition.
\end{proof}

\paragraph{The figure's prior and calculation.} The figure takes \((b_1,b_2)\) uniform on the baseline triangle and \(b_3\), conditional on \(b_2\), uniform on \([0,b_2]\). The resulting density is \(2/b_2\) on the ordered three-dimensional support. It is affiliated: the density factors by a single coordinate and the ordered support is a sublattice. Its baseline posterior remains \(\phi(c)=(1-c)/(1+c)\). With \(p_1=.9\), \(p_2=.1\), and \(d/\ell=.15\), the victim-payoff inequalities hold, \(\kappa=.25\), and \(\hat c=.6\). The curve solves the text's indifference condition using the exact piecewise-linear boundary of \(K_t(c)\). Horizontal and vertical integration independently give the same continuation-attack probabilities.

\section{An all-three-victim coalition example}\label{oa:three-victims}
This section gives the robustness example discussed in \Cref{sec:more-than-two}, allowing all three victims to participate and deriving their support decisions from incentives. Let \(\ell=1\), \(d=.3\), \((p_1,p_2,p_3)=(.9,.2,.1)\), and \(c=.03\). The three possible attacker types and their probabilities are
\[
(1,0,0):.4,\qquad (1,.5,0):.6(1-s),\qquad (1,.5,.5):.6s.
\]
Benefits are private. After \(V_1\) fights, \(V_2\) decides whether to join, then \(V_3\) decides whether to join the existing fight. Victim 3 can thus help victim 1 even when victim 2 abstains. A member of a coalition of size \(m\) pays \(d\) and faces conquest probability \(p_m\); the attacker pays \(mc\). Coalition members are not attacked again and cannot defend remaining outsiders later. Outsiders remain reachable after a pair war regardless of its outcome. Following surrender or a lone fight, the attacker attacks the next remaining victim in the fixed order or stops, ending the game; she cannot bypass \(V_2\) to reach \(V_3\). Victim 3 can join victim 2's later fight. Regularity rules out type inference from uninformed victims' deviations.

All types attack initially under every possible coalition response: even the least profitable triple fight yields \(.1(1)-3(.03)=.01>0\). If \(V_2\) is attacked later, Bayes' rule eliminates \((1,0,0)\), leaving probability \(s\) that \(V_3\) is also profitable. Both remaining types attack \(V_2\) even if she will receive help: pair-war payoffs are \(.04\) and \(.14\). Thus the downstream posterior accounts for the attacker's continuation incentives. Pair defense costs a victim \(.5\), triple defense \(.4\), and lone fighting \(1.2\).

At \(s=.2\), victim 3 refuses to help victim 2 later: staying out costs \(.2\), compared with \(.5\) for joining. Victim 2 consequently loses \(.6\) in expectation by waiting, compared with \(.5\) by defending victim 1 now. She joins. Victim 3 refuses the initial triple coalition, since her outside loss is \(.6s=.12<.4\). If victim 2 were to abstain, victim 3 would also refuse to defend victim 1 as a pair: waiting costs \(.12<.5\). Victim 1 anticipates pair support and fights.

At \(s=.8\), victim 3 helps victim 2 later: the outside loss \(.8\) exceeds the pair-defense loss \(.5\). Victim 2's expected loss from waiting is now \(.6\times.5=.3\), below either pair defense \(.5\) or triple defense \(.4\). She stays out initially. Victim 3 likewise declines to help victim 1 alone: immediate pair defense costs \(.5\), while waiting to defend victim 2 costs \(.3\) in expectation. Victim 1 therefore surrenders. Yet if victim 2 joined initially, victim 3 would join the triple coalition, since her outside loss \(.6s=.48\) exceeds \(.4\). All-three support is credible, but victim 2 prefers the protected future.

All participation and attack incentives are strict at both parameter values, so the conclusions survive small perturbations within the stated finite-support model. 

\end{document}